\documentclass[12pt,a4paper]{article}

\usepackage{iftex}
\ifPDFTeX
	\usepackage{cmap}
	\usepackage[T1]{fontenc}
	\usepackage[utf8]{inputenc}
	\usepackage[english]{babel}
\else
    \usepackage{fontspec}
    \usepackage[english]{babel}
\fi

\usepackage{amsmath,amssymb,amsfonts,amsthm,mathtools}

\usepackage{geometry}
\usepackage{booktabs}
\usepackage{array}
\usepackage{hyperref}
\hypersetup{
	colorlinks=true,
	linkcolor=blue,
	citecolor=blue,
	urlcolor=blue
}

\numberwithin{equation}{section}
\theoremstyle{plain}
\newtheorem{theorem}{Theorem}[section]
\newtheorem{lemma}{Lemma}[section]
\newtheorem{proposition}{Proposition}[section]

\theoremstyle{definition}
\newtheorem{definition}{Definition}[section]
\newtheorem{example}{Example}[section]
\newtheorem{algorithm}{Algorithm}[section]

\theoremstyle{remark}
\newtheorem{remark}{Remark}[section]

\DeclareMathOperator{\ISM}{ISM}
\DeclareMathOperator{\dist}{dist}
\DeclareMathOperator{\diam}{diam}
\DeclareMathOperator{\sgn}{sgn}

\DeclareMathOperator{\tridiag}{tridiag}
\DeclareMathOperator{\diag}{diag}
\DeclareMathOperator{\sn}{sn}
\DeclareMathOperator{\cn}{cn}
\DeclareMathOperator{\dn}{dn}

\title{\textbf{Optimization Inverse Spectral Problems for Sturm--Liouville Operators\\with Potentials in $L^p(0,1)$}}
\author{Nur~Valeev\\ \small \textit{Institute of Mathematics with Computing Centre, Ufa Federal Research Centre}\\[-0.2ex]
\small \textit{Russian Academy of Sciences}\\ \small \texttt{valeevnf@yandex.ru}}
\date{}

\begin{document}
	
	\maketitle
	
	\begin{abstract}
		We study an optimization inverse spectral problem (OISP) for a one-dimensional Sturm--Liouville operator on a finite interval. The problem is to find a potential $V \in L^p(0,1)$ that realizes a prescribed finite set $\lambda_1^* < \lambda_2^* < \dots < \lambda_m^*$ and has the smallest $L^p(0,1)$-distance from an a priori given potential $V_0 \in L^p(0,1)$.
		
		We formulate and prove results on the regularity of the isospectral manifold, the existence of an OISP solution, and an explicit analytic representation of the optimal potential by means of a power map applied to a linear combination of squared eigenfunctions.
		
		Using a unified differential-algebraic approach to squared eigenfunctions, we show that whenever these squares are linearly dependent, the corresponding potential can be expressed in terms of them and is an analytic function.
		
		For $p=2$, we prove that nondegenerate stationary points are isolated and that the number of global minima is finite provided that each of them is nondegenerate. For every $1<p<\infty$, we establish uniqueness for the first eigenvalue when $\lambda_1^*\ge\lambda_1(V_0)$ and generic global uniqueness in the sense of Baire category in $L^p(0,1)$.
		
		A separate section provides examples with at least two global minima for $m=1,2$ in the presence of a sufficiently high finite symmetric barrier. We describe an iteration with frozen eigenfunctions and a local Newton--KKT method with continuation in the spectral data, and state explicitly the limits of their guarantees.
		
		The final part discusses applications of OISP theory to nonlinear boundary value problems and to applied problems of diagnostics and optimal spectral design of acoustic and electromagnetic resonators.
		
		\medskip
		\noindent\textbf{Keywords:} optimization inverse spectral problem, Sturm--Liouville operator, isospectral manifold, metric projection in $L^p$, squared eigenfunctions, system of nonlinear Schr\"odinger equations, generic uniqueness, numerical methods, spectral design of resonators.
		
		\medskip
		\noindent\textbf{MSC 2020:} 34A55, 34L05, 34B15, 49J20, 49K20, 65L09.
	\end{abstract}
	
	\setcounter{section}{-1}
	
	\section{Introduction}
	
	The Sturm--Liouville operator (the one-dimensional stationary Schr\"odinger operator)
	\begin{equation}\label{eq:intro_op}
		\mathcal{L}[V]\psi := -\psi''(x) + V(x)\psi(x)
	\end{equation}
	is one of the fundamental models of mathematical physics, spectral theory, and vibration mechanics. In quantum mechanics, its eigenvalues are interpreted as energy levels of stationary states and its eigenfunctions as the corresponding wave functions; in vibration theory and acoustics, they describe squared natural frequencies and standing-wave modes in inhomogeneous media. Sturm--Liouville operators provide model examples for numerous problems of mathematical physics, including inverse spectral theory.
	
	The classical inverse spectral problem is to recover the potential $V(x)$ and the boundary-condition parameters from a sufficiently complete set of spectral characteristics \cite{Marchenko1977,Poschel1987,Gladwell2004}. Since the pioneering paper of V.~A.~Ambarzumian (1929) \cite{Ambarzumian1929} and the fundamental work of G.~Borg (1946) \cite{Borg1946}, it has been well known that, in general, a single spectrum does not determine the potential uniquely. Borg's uniqueness theorem uses two spectra corresponding to different boundary conditions at one endpoint of the interval \cite{Borg1946}. In the classical works of V.~A.~Marchenko \cite{Marchenko1977} and I.~M.~Gelfand and B.~M.~Levitan \cite{GelfandLevitan1951}, recovery of a potential from the spectral function, or from one spectrum together with a set of norming constants, was reduced to linear integral equations. The global geometry of spectral data and isospectral sets for the Dirichlet problem on an interval was studied in detail by J.~P\"oschel and E.~Trubowitz \cite{Poschel1987}. Completeness and admissibility of spectral data, as well as uniqueness and stability of potential recovery, have remained central questions of the classical theory.
	
	In the overwhelming majority of experimental, engineering, and computational problems, however, only a few of the first (lowest) eigenvalues or resonance frequencies $\lambda_1^*,\dots,\lambda_m^*$, with $m<\infty$, are available \cite{ChuGolub2005,Gladwell2004,IlyasovValeev2021}. In the classical formulation, recovering a potential from a finite set of eigenvalues is strongly underdetermined: the data are realized by an infinite-dimensional family of potentials, namely an isospectral set \cite{Poschel1987,IlyasovValeev2018UMJ,IlyasovValeev2021}. Various approaches to completing finite spectral data by fixing their high-frequency asymptotics or prescribing the potential a priori on part of the domain were considered, in particular, by B.~D.~Lowe, M.~Pilant, and W.~Rundell \cite{Lowe1992}. A related direction is formed by extremal spectral problems, originating in the works of J.~B.~Keller \cite{Keller1961} and M.~S.~Ashbaugh and E.~M.~Harrell \cite{Ashbaugh1987}, which study sharp estimates of eigenvalues and intermodal gaps on balls in $L^p$ and $L^1$ spaces; see Q.~Wei, G.~Meng, and M.~Zhang \cite{WeiMengZhang2009}, Y.~Tian \cite{TianWeiZhang2023}, and M.~Zhao and J.~Qi \cite{QiChen2016,GuoQi2017,ZhaoQi2025}.
	
	Applied inverse problems usually come with a priori information about the system in the form of a known approximate (reference or baseline) potential $V_0(x)$ \cite{ValeevIlyasov2018MZ,IlyasovValeev2021,Sadovnichii2026R3}. Requiring the unknown potential $V(x)$ to deviate minimally from the reference model $V_0(x)$ provides a natural and physically motivated variational principle for selecting one solution among all potentials having the observed part of the spectrum. This leads to the general concept of an \textit{optimization inverse spectral problem (OISP)}, also referred to as an \textit{inverse optimal problem (IOP)} \cite{ValeevIlyasov2018MZ,IlyasovValeev2021,Valeev2026Graph}: given a reference element $V_0$ and a finite collection of spectral measurements $\Lambda^*$, one seeks a potential $\hat V$ minimizing the distance $\|V-V_0\|$ over all admissible potentials whose spectra agree with the data $\Lambda^*$.
	
	The systematic development of OISP theory and its connections with nonlinear analysis was initiated in the works of N.~F.~Valeev and Y.~Sh.~Ilyasov \cite{ValeevIlyasov2018MZ,IlyasovValeev2018UMJ,IlyasovValeev2019JDE,ValeevIlyasov2020UMJ,IlyasovValeev2021}, and in a series of joint studies by V.~A.~Sadovnichii, Ya.~T.~Sultanaev, and N.~F.~Valeev \cite{Sadovnichii2022DE,Sadovnichii2023Dokl,Sadovnichii2024DE,Sadovnichii2026R3}:
	\begin{itemize}
		\item For a one-dimensional Sturm--Liouville operator on an interval, \cite{ValeevIlyasov2018MZ,IlyasovValeev2018UMJ} established solvability of a one-parameter OISP and revealed its fundamental connection with the nonlinear boundary value equation $-u''+V_0(x)u=\lambda u+\sigma u^3$, $\sigma=\pm1$. This made it possible to prove a generalized Sturm nodal theorem for nonlinear boundary value problems.
		\item In \cite{IlyasovValeev2019JDE}, the theory was extended to a multidimensional Schr\"odinger operator on a bounded domain $\Omega\subset\mathbb R^N$. The paper \cite{IlyasovValeev2021} gave a complete solution of the $m$-parameter OISP in $L^2(-l,l)$ for an arbitrary finite set of eigenvalues $\lambda_{i_1}<\dots<\lambda_{i_m}$ and related the necessary optimality conditions to a stationary system of $m$ nonlinearly coupled Schr\"odinger equations.
		\item In works by V.~A.~Sadovnichii, Ya.~T.~Sultanaev, and N.~F.~Valeev, the OISP method was developed for matrix and vector Sturm--Liouville operators \cite{Sadovnichii2022DE,Sadovnichii2023Dokl}, for the one-dimensional Schr\"odinger operator on the whole real line \cite{Sadovnichii2024DE}, and for the three-dimensional Schr\"odinger operator in $L^2(\mathbb R^3)$ \cite{Sadovnichii2026R3}.
		\item N.~F.~Valeev \cite{Valeev2026Graph} developed OISP theory for the Schr\"odinger operator on a compact metric graph. That paper proves nonemptiness of the isospectral set for every strictly increasing set of the first $m$ eigenvalues, represents the optimal potential through squared eigenfunctions on cycle-free graphs (trees), studies linear dependence of squared eigenfunctions on graphs with cycles using the Lam\'e potential on a circle, and obtains a limiting representation at nonregular points.
		\item The development of the OISP method also yielded new results on solvability of problems with a prescribed partial trace \cite{ValeevIlyasov2020UMJ}, existence of solutions to degenerate $(p,q)$-Laplacian equations \cite{IlyasovValeev2025BLMS}, and uniqueness of principal solutions to defocusing Hartree equations \cite{IlyasovSunValeevYao2026}.
	\end{itemize}
	
	This paper studies properties of the OISP for Sturm--Liouville operators on $(0,1)$ throughout the Banach scale $L^p(0,1)$, $1<p<\infty$. Passing beyond the Hilbert space $L^2(0,1)$ is of theoretical interest, because orthogonality to the tangent space of the isospectral manifold is replaced by a nonlinear duality map in a Banach space. It is also of direct practical relevance when modeling restrictions on the peak amplitude, $p\gg2$, or the spatial localization, $1<p<2$, of a correction to the potential. We also investigate isolation of nondegenerate solutions for $p=2$, generic global uniqueness for all $1<p<\infty$, and mechanisms of nonuniqueness. Numerical schemes are considered with their local applicability conditions stated explicitly.
	
	\section{Formulation of the OISP with $V_0(x)\in L^p(0,1)$}
	
	\subsection{The direct spectral problem and the isospectral set}
	
	On the finite interval $(0,1)$, consider the self-adjoint Sturm--Liouville operator $\mathcal L[V]$ generated by the differential expression
	\begin{equation}\label{eq:sl_op}
		\mathcal{L}[V]\psi := -\psi''(x) + V(x)\psi(x), \quad x \in (0,1),
	\end{equation}
	with Dirichlet boundary conditions at the endpoints:
	\begin{equation}\label{eq:dirichlet_bc}
		\psi(0) = 0, \quad \psi(1) = 0.
	\end{equation}
	\begin{remark}
		All subsequent results extend, with straightforward technical modifications, to general separated self-adjoint boundary conditions $\psi(0)\cos\alpha+\psi'(0)\sin\alpha=0$ and $\psi(1)\cos\beta+\psi'(1)\sin\beta=0$ considered in \cite{ValeevIlyasov2018MZ,IlyasovValeev2018UMJ}.
	\end{remark}
	
	Let $1<p<\infty$, and let the potential $V$ belong to the real Banach space $L^p(0,1):=L^p(0,1;\mathbb R)$ endowed with the norm
	\[
	\|V\|_{L^p} = \left( \int_0^1 |V(x)|^p \, dx \right)^{1/p}.
	\]
	We write $q=\frac{p}{p-1}\in(1,\infty)$ for the exponent conjugate to $p$, so that $1/p+1/q=1$. Since $L^p(0,1)\hookrightarrow L^1(0,1)$ for every $p>1$, the quadratic form
	\begin{equation}\label{eq:quad_form}
		a_V[u, u] = \int_0^1 |u'(x)|^2 \, dx + \int_0^1 V(x)|u(x)|^2 \, dx, \quad u \in W_0^{1,2}(0,1),
	\end{equation}
	is bounded from below and closed on the Sobolev space $W_0^{1,2}(0,1)$ \cite{IlyasovValeev2021,Valeev2026Graph}. The domain of $\mathcal L[V]$ is
	\[
	\mathcal{D}(\mathcal{L}[V]) = \left\{ u \in W_0^{1,2}(0,1) : -u'' + Vu \in L^2(0,1)\ \text{in the sense of distributions} \right\},
	\]
	For $p\ge2$, this domain equals $H^2(0,1)\cap H_0^1(0,1)$. The eigenvalue equation gives $W^{2,p}(0,1)$ regularity for every eigenfunction and every $p>1$. The spectrum is purely discrete and consists of an infinite strictly increasing sequence of simple real eigenvalues \cite{Marchenko1977,IlyasovValeev2021}:
	\begin{equation}\label{eq:spectrum_chain}
		\lambda_1(V) < \lambda_2(V) < \dots < \lambda_k(V) < \dots \to +\infty \quad (k \to \infty).
	\end{equation}
	Each eigenvalue $\lambda_k(V)$ corresponds to a real eigenfunction $\varphi_k(V)(x)$, unique up to sign and normalized in $L^2(0,1)$ by
	\begin{equation}\label{eq:ef_norm}
		\|\varphi_k(V)\|_{L^2(0,1)}^2 = \int_0^1 \varphi_k^2(V)(x) \, dx = 1, \quad \varphi_k'(V)(0) > 0,
	\end{equation}
	which, by Sturm's oscillation theorem, has exactly $k-1$ simple zeros (nodes) in $(0,1)$.
	
	Fix an integer $m\ge1$, indices $1\le i_1<i_2<\dots<i_m$ (for the most important case of the first $m$ eigenvalues, set $i_k=k$, $k=1,\dots,m$), and the vector of target spectral data
	\begin{equation}\label{eq:omega_m}
		\Lambda^* = (\lambda_1^*, \lambda_2^*, \dots, \lambda_m^*)^T \in \Omega_m := \left\{ (\mu_1, \dots, \mu_m)^T \in \mathbb{R}^m : \mu_1 < \mu_2 < \dots < \mu_m \right\}.
	\end{equation}
	Define the spectral map $F_m:L^p(0,1)\to\mathbb R^m$ by $F_m(V)=(\lambda_{i_1}(V),\dots,\lambda_{i_m}(V))^T$, and introduce the \textit{isospectral set} $\ISM_m^p(\Lambda^*)$ in $L^p(0,1)$ \cite{IlyasovValeev2021,Valeev2026Graph}:
	\begin{equation}\label{eq:ism_def}
		\ISM_m^p(\Lambda^*) := \left\{ V \in L^p(0,1) : \lambda_{i_k}(V) = \lambda_k^*, \;\; k = 1, \dots, m \right\}.
	\end{equation}
	
	\subsection{Formulation of the OISP}
	
	Let a reference (a priori) potential $V_0\in L^p(0,1)$ be given.
	
	The \textit{optimization inverse spectral problem $\mathcal P_p(m)$} is to find a potential $\hat V\in\ISM_m^p(\Lambda^*)$ that realizes the distance in $L^p(0,1)$ from $V_0$ to the isospectral set $\ISM_m^p(\Lambda^*)$ \cite{IlyasovValeev2021,Valeev2026Graph}:
	\begin{equation}\label{eq:oisp_dist}
		\|\hat{V} - V_0\|_{L^p(0,1)} = \dist_{L^p}\big(V_0, \ISM_m^p(\Lambda^*)\big) := \inf_{V \in \ISM_m^p(\Lambda^*)} \|V - V_0\|_{L^p(0,1)}.
	\end{equation}
	Since $s\mapsto\frac1p s^p$ is strictly increasing for $s\ge0$, problem \eqref{eq:oisp_dist} is equivalent to the smooth variational problem with equality constraints
	\begin{equation}\label{eq:oisp_var}
		J_p(V) := \frac{1}{p} \int_0^1 |V(x) - V_0(x)|^p \, dx \longrightarrow \min, \quad V \in \ISM_m^p(\Lambda^*).
	\end{equation}
	
	\subsection{The case $p=2$}
	
	For $p=2$, the potential space $X=L^2(0,1)$ is a Hilbert space with inner product $\langle f,g\rangle=\int_0^1 f(x)g(x)\,dx$. We recall the principal results established for this case in \cite{ValeevIlyasov2018MZ,IlyasovValeev2018UMJ,IlyasovValeev2021,Valeev2026Graph}.
	
	\begin{proposition}[Nonemptiness and smooth structure of $\ISM_m^2(\Lambda^*)$ \cite{IlyasovValeev2021,Valeev2026Graph}]\label{prop:p2_manifold}
		Let $\Lambda^*=(\lambda_1^*,\dots,\lambda_m^*)^T\in\Omega_m$ and $1\le i_1<\dots<i_m$. Then:
		\begin{enumerate}
			\item The set $\ISM_m^2(\Lambda^*)$ is nonempty and contains bounded potentials $V\in L^\infty(0,1)$.
			\item The functionals $\lambda_k:L^2(0,1)\to\mathbb R$ are continuously Fr\'echet differentiable, $\lambda_k\in C^1(L^2(0,1))$, and
			\begin{equation}\label{eq:frechet_p2}
				D\lambda_k(V)[h] = \langle \varphi_k^2(V), h \rangle = \int_0^1 \varphi_k^2(V)(x) h(x) \, dx, \quad \forall h \in L^2(0,1).
			\end{equation}
			\item For every potential $V\in L^2(0,1)$, the squared eigenfunctions $\{\varphi_{i_1}^2(V),\dots,\varphi_{i_m}^2(V)\}$ are linearly independent in $L^2(0,1)$, and $\ISM_m^2(\Lambda^*)$ is a smooth submanifold of codimension $m$ in $L^2(0,1)$.
		\end{enumerate}
	\end{proposition}
	
	\begin{theorem}[Existence and representation of an OISP solution for $p=2$ \cite{IlyasovValeev2021}]\label{thm:p2_summary}
		For every reference potential $V_0\in L^2(0,1)$ and every $\Lambda^*\in\Omega_m$, problem $\mathcal P_2(m)$ in \eqref{eq:oisp_var} has at least one global solution $\hat V\in\ISM_m^2(\Lambda^*)$. At every local or global minimum $\hat V$, there exist uniquely determined Lagrange multipliers $\alpha_1,\dots,\alpha_m\in\mathbb R$ such that
		\begin{equation}\label{eq:p2_rep}
			\hat{V}(x) = V_0(x) - \sum_{k=1}^m \alpha_k \varphi_{i_k}^2(\hat{V})(x) \quad \text{a.e. on } (0,1).
		\end{equation}
		If $\sigma_k=\sgn(\alpha_k)\in\{-1,0,+1\}$ and $\hat u_k(x)=\sqrt{|\alpha_k|}\,\varphi_{i_k}(\hat V)(x)$ (with $\hat u_k=0$ when $\alpha_k=0$), then $(\hat u_1,\dots,\hat u_m)\in(W^{2,2}(0,1)\cap W_0^{1,2}(0,1))^m\subset(C^{3/2}[0,1])^m$ satisfies the stationary system of coupled nonlinear Schr\"odinger equations
		\begin{equation}\label{eq:p2_nls}
			\left\{ \begin{aligned}
				-\hat{u}_k''(x) + V_0(x)\hat{u}_k(x) &= \lambda_k^* \hat{u}_k(x) + \left( \sum_{j=1}^m \sigma_j \hat{u}_j^2(x) \right) \hat{u}_k(x), \quad x \in (0,1), \\
				\hat{u}_k(0) = \hat{u}_k(1) &= 0, \quad k = 1, \dots, m,
			\end{aligned} \right.
		\end{equation}
		The functions $\{\hat u_k\}_{k=1}^m$ are pairwise orthogonal in $L^2(0,1)$; whenever $\alpha_k\ne0$, the function $\hat u_k(x)$ has exactly $i_k-1$ simple zeros in $(0,1)$; and the optimal potential is recovered from
		\begin{equation}\label{eq:p2_pot_u}
			\hat{V}(x) = V_0(x) - \sum_{k=1}^m \sigma_k \hat{u}_k^2(x).
		\end{equation}
	\end{theorem}
	
	\section{The OISP for the Schr\"odinger operator with a potential in $L^p(0,1)$, $1<p<\infty$}
	
	We now study the OISP in the Banach space $L^p(0,1)$ for arbitrary $1<p<\infty$. Unlike the case $p=2$, the space $L^p(0,1)$ is not Hilbert when $p\ne2$: its dual $(L^p(0,1))^*$ is isometrically isomorphic to $L^q(0,1)$, where $q=\frac{p}{p-1}$, and metric projection relates the primal and dual spaces through a nonlinear duality map.
	
	\subsection{Auxiliary continuity and regularity lemmas in $L^p(0,1)$}
	
	\begin{lemma}[Weak continuity of spectral data in $L^p(0,1)$]\label{lem:weak_cont_lp}
		Let $1<p<\infty$, and suppose that a sequence of potentials $\{V_j\}_{j=1}^\infty\subset L^p(0,1)$ converges weakly to $V\in L^p(0,1)$ as $j\to\infty$, that is, $V_j\rightharpoonup V$ in $L^p(0,1)$. Then, for every fixed $k\ge1$,
		\begin{equation}\label{eq:lem_weak_conv}
			\lim_{j \to \infty} \lambda_k(V_j) = \lambda_k(V), \quad \lim_{j \to \infty} \|\varphi_k(V_j) - \varphi_k(V)\|_{C^1[0,1]} = 0.
		\end{equation}
	\end{lemma}
	
	\begin{proof}
		1. Endow the form domain $Q=W_0^{1,2}(0,1)$ with the standard norm $\|u\|_Q=\left(\int_0^1(|u'|^2+|u|^2)\,dx\right)^{1/2}$. By the one-dimensional Sobolev embedding theorem and the Arzel\`a--Ascoli theorem, the embedding $Q\hookrightarrow C[0,1]$ is compact, and $\|u\|_{C[0,1]}\le a\|u\|_Q$ for all $u\in Q$.
		
		2. Set $w_j=V_j-V$. Since $w_j\rightharpoonup0$ weakly in $L^p(0,1)$, the sequence $\|w_j\|_{L^p}$ is bounded. We claim that
		\begin{equation}\label{eq:eps_j_zero}
			\varepsilon_j := \sup_{\|u\|_Q \le 1} \left| \int_0^1 (V_j(x) - V(x)) |u(x)|^2 \, dx \right| \longrightarrow 0 \quad (j \to \infty).
		\end{equation}
		Assume the contrary. Then there are $\delta>0$, a subsequence (still indexed by $j$), and $u_j\in Q$ with $\|u_j\|_Q\le1$ such that $\left|\int_0^1w_j|u_j|^2dx\right|\ge\delta$. By compactness of $Q\hookrightarrow C[0,1]$, after passing to a subsequence we may assume that $u_j\to u_*$ strongly in $C[0,1]$. Hence $|u_j|^2\to|u_*|^2$ uniformly on $[0,1]$ and therefore strongly in $L^q(0,1)$, where $q=\frac{p}{p-1}$. H\"older's inequality gives
		\[
		\left| \int_0^1 w_j |u_j|^2 \, dx \right| \le \|w_j\|_{L^p} \left\| |u_j|^2 - |u_*|^2 \right\|_{L^q} + \left| \int_0^1 w_j(x) |u_*(x)|^2 \, dx \right|.
		\]
		The first term on the right tends to zero because $\|w_j\|_{L^p}$ is bounded and $|u_j|^2\to|u_*|^2$ strongly in $L^q(0,1)$. The second tends to zero by the definition of weak convergence $w_j\rightharpoonup0$ in $L^p(0,1)$, since $|u_*|^2\in C[0,1]\subset L^q(0,1)$. This contradiction proves \eqref{eq:eps_j_zero}.
		
		3. Using the Gagliardo--Nirenberg inequality $\|u\|_{C[0,1]}^2\le2\|u\|_{L^2}\|u'\|_{L^2}\le\eta\|u'\|_{L^2}^2+\eta^{-1}\|u\|_{L^2}^2$ and H\"older's inequality $\left|\int_0^1V|u|^2dx\right|\le\|V\|_{L^p}\|u\|_{C[0,1]}^2$, choose a shift $c>0$, independent of $j$, so that the shifted form $b_V[u]:=a_V[u,u]+c\|u\|_{L^2}^2$ satisfies the coercivity estimate $b_V[u]\ge\frac12\|u\|_Q^2$ for all $u\in Q$. Then \eqref{eq:eps_j_zero} implies, for every $u\in Q$,
		\begin{equation}\label{eq:form_sandwich}
			(1 - 2\varepsilon_j) b_V[u] \le b_{V_j}[u] \le (1 + 2\varepsilon_j) b_V[u].
		\end{equation}
		Applying the Courant--Fischer min--max principle to \eqref{eq:form_sandwich},
		\[
		\lambda_k(V) + c = \inf_{\substack{E \subset Q \\ \dim E = k}} \sup_{u \in E \setminus \{0\}} \frac{b_V[u]}{\|u\|_{L^2}^2},
		\]
		we obtain, for all sufficiently large $j$, when $2\varepsilon_j<1$, the two-sided estimate
		\[
		(1 - 2\varepsilon_j)\big(\lambda_k(V) + c\big) \le \lambda_k(V_j) + c \le (1 + 2\varepsilon_j)\big(\lambda_k(V) + c\big).
		\]
		Subtracting $\lambda_k(V)+c$ yields
		\[
		|\lambda_k(V_j) - \lambda_k(V)| \le 2\varepsilon_j \big(\lambda_k(V) + c\big) \longrightarrow 0 \quad (j \to \infty).
		\]
		
		4. Uniform boundedness of $\|V_j\|_p$ and of the eigenvalues, together with the form estimate in Step~3, implies uniform boundedness of $\|\varphi_k(V_j)\|_{H_0^1}$ and hence of $\|\varphi_k(V_j)\|_C$. The equations $-\varphi_k''(V_j)=(\lambda_k(V_j)-V_j)\varphi_k(V_j)$ and the bound $\|\varphi_k(V_j)\|_{C[0,1]}\le C_k$ imply $\|\varphi_k''(V_j)\|_{L^p(0,1)}\le\widetilde C_k$. Since $W^{2,p}(0,1)\hookrightarrow C^1[0,1]$ compactly for $p>1$, a subsequence of $\{\varphi_k(V_j)\}$ converges strongly in $C^1[0,1]$ to a solution $\varphi_k^*$ of the limiting problem $\mathcal L[V]\varphi_k^*=\lambda_k(V)\varphi_k^*$, with $\|\varphi_k^*\|_{L^2}=1$ and $\varphi_k^{*\prime}(0)>0$. The eigenfunction with this normalization is unique; therefore the whole sequence $\varphi_k(V_j)$ converges to $\varphi_k(V)$ in $C^1[0,1]$.
	\end{proof}
	
	\begin{lemma}[Fr\'echet differentiability and linear independence of squared eigenfunctions in $L^p(0,1)$]\label{lem:frechet_indep_lp}
		Let $1<p<\infty$ and $q=\frac{p}{p-1}$. Then:
		\begin{enumerate}
			\item For every $k\ge1$, the functional $\lambda_k:L^p(0,1)\to\mathbb R$ is real analytic (in particular, it belongs to $C^\infty(L^p(0,1))$), and its Fr\'echet derivative at $V\in L^p(0,1)$ is represented by $\varphi_k^2(V)\in C[0,1]\subset L^q(0,1)$ as follows:
			\begin{equation}\label{eq:frechet_lp}
				\langle D\lambda_k(V), h \rangle = \int_0^1 \varphi_k^2(V)(x) h(x) \, dx, \quad \forall h \in L^p(0,1).
			\end{equation}
			\item For every potential $V\in L^p(0,1)$ and every finite set of indices $1\le i_1<i_2<\dots<i_m$, the functions $\{\varphi_{i_1}^2(V),\dots,\varphi_{i_m}^2(V)\}$ are linearly independent on $(0,1)$.
		\end{enumerate}
	\end{lemma}
	Consequently, for every $\Lambda^*\in\Omega_m$, the isospectral set $\ISM_m^p(\Lambda^*)$ is a smooth real-analytic submanifold of codimension $m$ in the Banach space $L^p(0,1)$.
	\begin{proof}	
		1. For $u\in W_0^{1,2}(0,1)$,
		$\left|\int_0^1h|u|^2dx\right|\le\|h\|_{L^p}\|u\|_{C[0,1]}^2\le\|h\|_{L^p}(\eta\|u'\|_{L^2}^2+\eta^{-1}\|u\|_{L^2}^2)$.
		Thus, for every perturbation $h\in L^p(0,1)$, multiplication by $h$ is a relatively form-bounded perturbation of $\mathcal L[V]$. All eigenvalues $\lambda_k(V)$ of the Dirichlet problem on $(0,1)$ are simple; hence Kato's analytic perturbation theory \cite{Kato1995} yields real analyticity of the map $V\mapsto\lambda_k(V)$ from $L^p(0,1)$ to $\mathbb R$ and the first-variation formula \eqref{eq:frechet_lp}.
		
		2. We prove the linear independence of $\varphi_{i_1}^2(V),\dots,\varphi_{i_m}^2(V)$ for $V\in L^p(0,1)$ by developing the method of \cite{IlyasovValeev2021,Valeev2026Graph}. Suppose that, for some real constants $c_1,\dots,c_m\in\mathbb R$,
		\begin{equation}\label{eq:lin_dep_assump}
			\sum_{k=1}^m c_k \varphi_{i_k}^2(x) = 0, \quad x \in [0,1].
		\end{equation}
		Since $\varphi_{i_k}\in W^{2,p}(0,1)\hookrightarrow C^1[0,1]$, differentiate \eqref{eq:lin_dep_assump} twice on $(0,1)$:
		\begin{equation}\label{eq:lin_dep_diff12}
			\sum_{k=1}^m c_k \varphi_{i_k}(x)\varphi_{i_k}'(x) = 0, \quad \sum_{k=1}^m c_k \left( (\varphi_{i_k}'(x))^2 + \varphi_{i_k}(x)\varphi_{i_k}''(x) \right) = 0.
		\end{equation}
		Substituting $\varphi_{i_k}''=(V(x)-\mu_k)\varphi_{i_k}$, where $\mu_k:=\lambda_{i_k}(V)$, into the second equality in \eqref{eq:lin_dep_diff12} and using \eqref{eq:lin_dep_assump}, we obtain
		\begin{equation}\label{eq:lin_dep_prime_sq}
			\sum_{k=1}^m c_k (\varphi_{i_k}'(x))^2 = \sum_{k=1}^m c_k \mu_k \varphi_{i_k}^2(x), \quad x \in [0,1].
		\end{equation}
		Since $\varphi_{i_k}\in W^{2,p}(0,1)$, the function $(\varphi_{i_k}')^2$ belongs to $W^{1,p}(0,1)$ and its derivative is $((\varphi_{i_k}')^2)'=2\varphi_{i_k}'\varphi_{i_k}''=2(V(x)-\mu_k)\varphi_{i_k}\varphi_{i_k}'$ almost everywhere. Differentiating \eqref{eq:lin_dep_prime_sq} and using the first identity in \eqref{eq:lin_dep_diff12}, we find
		\[
		-2\sum_{k=1}^m c_k \mu_k \varphi_{i_k}(x)\varphi_{i_k}'(x) = 2\sum_{k=1}^m c_k \mu_k \varphi_{i_k}(x)\varphi_{i_k}'(x) \implies \sum_{k=1}^m c_k \mu_k \varphi_{i_k}(x)\varphi_{i_k}'(x) = 0.
		\]
		Integrating from $0$ to $x$ and using the Dirichlet condition $\varphi_{i_k}(0)=0$, we obtain $\sum_{k=1}^m c_k\mu_k\varphi_{i_k}^2(x)=0$ on $[0,1]$. By \eqref{eq:lin_dep_prime_sq}, this also gives $\sum_{k=1}^m c_k(\varphi_{i_k}'(x))^2=0$ on $[0,1]$. Iterating this argument with respect to the power $s=0,1,\dots,m-1$ yields, at $x=0$, the system
		\begin{equation}\label{eq:vandermonde_sys}
			\sum_{k=1}^m \mu_k^s \left( c_k (\varphi_{i_k}'(0))^2 \right) = 0, \quad s = 0, 1, \dots, m-1.
		\end{equation}
		Since the eigenvalues $\mu_1<\mu_2<\dots<\mu_m$ are pairwise distinct, the Vandermonde determinant in \eqref{eq:vandermonde_sys} is nonzero. Hence $c_k(\varphi_{i_k}'(0))^2=0$ for all $k=1,\dots,m$. By uniqueness for the Cauchy problem, a nontrivial eigenfunction satisfying $\varphi_{i_k}(0)=0$ must have $\varphi_{i_k}'(0)\ne0$. Therefore $c_1=c_2=\dots=c_m=0$.
		
		Linear independence of the functionals $D\lambda_{i_k}(V)=\varphi_{i_k}^2(V)\in L^q(0,1)$ means that $DF_m(V):L^p(0,1)\to\mathbb R^m$ is surjective and has a complemented kernel of codimension $m$. The implicit function theorem in Banach spaces therefore proves the asserted smooth manifold structure of $\ISM_m^p(\Lambda^*)$.
	\end{proof}
	
	\subsection{Linear dependence and independence of squared eigenfunctions and analyticity of the potential}
	
	Linear independence or dependence of squared eigenfunctions of differential operators plays an important role in spectral theory, inverse problems, the calculus of variations, and optimal control. By the first-variation formula $\langle D\lambda_k(V),h\rangle=\int\varphi_k^2(V)h\,dx$, the squares of normalized eigenfunctions are gradients (Fr\'echet derivatives) of the eigenvalues with respect to the coefficients of the operator \cite{Borg1946,Poschel1987}. Hence linear independence of a finite or countable family of squared eigenfunctions determines the local geometry of spectral maps, guarantees regularity of isospectral sets as smooth submanifolds---and thereby provides a description of their tangent and normal bundles in the classical inverse spectral problem \cite{Poschel1987}---and is the necessary and sufficient constraint qualification, i.e. surjectivity of $DF_m(V)$, for deriving Euler--Lagrange equations in optimization inverse spectral problems \cite{IlyasovValeev2021,Valeev2026Graph}.
	
	In the theory of extremal spectral problems for Sturm--Liouville operators, linear dependence of squared eigenfunctions was singled out and studied by T.~J.~Mahar and B.~E.~Willner \cite{MaharWillner1976,MaharWillner1980}. They showed that such dependence, which may occur for nonseparated and periodic boundary conditions, leads to degeneracy of the optimality conditions and to abnormal extremals. In multidimensional problems and control theory for partial differential equations, linear independence of squared eigenfunctions, and in particular the basis properties generated by them, underlies optimal placement of sensors and actuators and bilinear control of quantum systems. For example, Y.~Privat and M.~Sigalotti \cite{PrivatSigalotti2010} proved that linear independence of the squares of Dirichlet--Laplacian eigenfunctions is generic under variations of the domain geometry. An analogous genericity result under variations of the potential on metric graphs with cycles was obtained in \cite{Valeev2026Graph}.
	
	In the proof of part~2 of Lemma~\ref{lem:frechet_indep_lp}, the key step was repeated differentiation of a linear combination of squared eigenfunctions and cancellation of the resulting constants of integration at the endpoint $x=0$ by the Dirichlet condition. We now show that this procedure is a special case of a general differential-algebraic scheme. Within a unified framework, this scheme proves linear independence of squared eigenfunctions for general separated boundary conditions and proves analyticity of the potential $V(x)$ when the squared eigenfunctions are linearly dependent, as may happen, for example, under periodic boundary conditions \cite{MaharWillner1980,Valeev2026Graph}.
	
	Let $V\in L^p(0,1)$, $1<p<\infty$, and let $u_1,\dots,u_m\in W^{2,p}(0,1)$, $u_k\not\equiv0$, be real solutions of the Sturm--Liouville equations
	\begin{equation}\label{eq:sl_gen_uk}
		-u_k''(x) + V(x)u_k(x) = \lambda_k u_k(x) \quad \text{a.e. on } (0,1), \qquad k = 1, \dots, m,
	\end{equation}
	corresponding to pairwise distinct real spectral parameters $\lambda_1,\dots,\lambda_m\in\mathbb R$, with $\lambda_i\ne\lambda_j$ for $i\ne j$. No boundary conditions are imposed on the functions $u_k$ a priori.
	
	For arbitrary real coefficients $c_1,\dots,c_m\in\mathbb R$ and every integer $r\ge0$, introduce two auxiliary functions:
	\begin{equation}\label{eq:Sr_Tr_def}
		S_r(x) := \sum_{k=1}^m c_k \lambda_k^r u_k^2(x), \qquad T_r(x) := \sum_{k=1}^m c_k \lambda_k^r \big(u_k'(x)\big)^2, \qquad x \in [0,1].
	\end{equation}
	
	\begin{lemma}[Fundamental recurrence identities]\label{lem:Sr_Tr_identities}
		Let $V\in L^p(0,1)$, $1<p<\infty$, and let $u_k\in W^{2,p}(0,1)$ satisfy \eqref{eq:sl_gen_uk}. Then, for every $r\ge0$, the functions $S_r\in W^{2,p}(0,1)\hookrightarrow C^1[0,1]$ and $T_r\in W^{1,p}(0,1)\hookrightarrow C[0,1]$ satisfy, in the sense of weak derivatives and almost everywhere on $(0,1)$, the recurrence identities
		\begin{equation}\label{eq:star_identities}
			S_r'' = 2T_r + 2V S_r - 2S_{r+1}, \qquad T_r' = V S_r' - S_{r+1}'. \tag{$*$}
		\end{equation}
		\stepcounter{equation}
		In particular, if $S_r(x)\equiv0$ on $[0,1]$ for some $r\ge0$, then $S_{r+1}(x)$ and $T_r(x)$ are identically equal to the same constant on $[0,1]$:
		\begin{equation}\label{eq:Sr_plus_1_const}
			S_{r+1}(x) = T_r(x) \equiv C_r = \mathrm{const}, \qquad x \in [0,1].
		\end{equation}
	\end{lemma}
	
	\begin{proof}
		Since $u_k\in W^{2,p}(0,1)\hookrightarrow C^1[0,1]$, we have $u_k^2\in W^{2,p}(0,1)$ and $(u_k')^2\in W^{1,p}(0,1)$. Differentiating $S_r$ and $T_r$ and substituting $u_k''=(V-\lambda_k)u_k$ from \eqref{eq:sl_gen_uk}, we obtain almost everywhere on $(0,1)$
		\[
		S_r' = 2\sum_{k=1}^m c_k \lambda_k^r u_k u_k',\]
		\[S_r'' = 2\sum_{k=1}^m c_k \lambda_k^r \big((u_k')^2 + u_k u_k''\big) = 2T_r + 2\sum_{k=1}^m c_k \lambda_k^r (V - \lambda_k)u_k^2 = 2T_r + 2V S_r - 2S_{r+1},
		\]
		\[
		T_r' = 2\sum_{k=1}^m c_k \lambda_k^r u_k' u_k'' = 2\sum_{k=1}^m c_k \lambda_k^r (V - \lambda_k) u_k u_k' = V S_r' - S_{r+1}',
		\]
		which proves \eqref{eq:star_identities}.
		
		Now suppose that $S_r(x)\equiv0$ on $[0,1]$. Then $S_r'\equiv0$ and $S_r''\equiv0$, and the first identity in \eqref{eq:star_identities} gives $T_r(x)=S_{r+1}(x)$ on $[0,1]$ by continuity. Weak differentiation of $T_r=S_{r+1}$ gives $T_r'=S_{r+1}'$, whereas the second identity in \eqref{eq:star_identities}, with $S_r'=0$, gives $T_r'=-S_{r+1}'$. Thus $2S_{r+1}'=0$ almost everywhere on $(0,1)$. Since $S_{r+1}\in W^{1,p}(0,1)$ is absolutely continuous, it is constant on $[0,1]$, and \eqref{eq:Sr_plus_1_const} follows.
	\end{proof}
	
	Lemma~\ref{lem:Sr_Tr_identities} shows that if the family $\{u_k^2\}_{k=1}^m$ satisfies $\sum_{k=1}^m c_ku_k^2(x)=0$, that is, $S_0(x)\equiv0$, then the subsequent structure is determined entirely by the constants $C_r$:
	\begin{itemize}
		\item \textbf{Case 1 (separated boundary conditions at one endpoint).} Suppose that, at one endpoint $x_*\in\{0,1\}$, all $u_k$ satisfy either the same Dirichlet condition $u_k(x_*)=0$ or the same Robin condition $u_k'(x_*)=\beta u_k(x_*)$, $\beta\in\mathbb R$. If $S_r\equiv0$, then in the Dirichlet case $C_r=S_{r+1}(x_*)=0$, while in the Robin case $C_r=T_r(x_*)=\beta^2S_r(x_*)=0$. Induction in $r$ gives $S_0(x)=S_1(x)=\dots=S_{m-1}(x)=0$ on $[0,1]$. Since the Vandermonde determinant $\det(\lambda_k^r)_{r,k=0}^{m-1}=\prod_{1\le i<j\le m}(\lambda_j-\lambda_i)$ is nonzero, it follows that $c_ku_k^2(x)\equiv0$. As $u_k\not\equiv0$, all $c_k$ vanish. This proves linear independence of the squared eigenfunctions (part~2 of Lemma~\ref{lem:frechet_indep_lp} and Theorem~4.1 in \cite{Valeev2026Graph}).
		\item \textbf{Case 2 (linear dependence of squared eigenfunctions without endpoint cancellation).} For periodic boundary conditions, or on edges of graphs containing cycles, the endpoint mechanism that forces the constants $C_r$ to vanish is absent. The first nonzero constant makes it possible to express the potential in terms of squared eigenfunctions and to prove its real analyticity.
	\end{itemize}
	
	\begin{theorem}[Analyticity of the potential under linear dependence of squared eigenfunctions]\label{thm:analytic_pot_dep}
		Let $V\in L^p(0,1)$, $1<p<\infty$, and let the nontrivial real solutions $u_1,\dots,u_m\in W^{2,p}(0,1)$, $u_k\not\equiv0$, of \eqref{eq:sl_gen_uk}, corresponding to pairwise distinct eigenvalues $\lambda_1,\dots,\lambda_m\in\mathbb R$, have \textbf{linearly dependent} squares $u_1^2,\dots,u_m^2$ on $(0,1)$:
		\begin{equation}\label{eq:uk_sq_dep}
			\sum_{k=1}^m c_k u_k^2(x) = 0, \qquad x \in [0,1],
		\end{equation}
		where the coefficients $c_1,\dots,c_m\in\mathbb R$ are not all zero. Then:
		\begin{enumerate}
			\item There exist an integer $r_0\in\{1,\dots,m-1\}$ and real constants $C\ne0$, $A\in\mathbb R$, and $b_1,\dots,b_m\in\mathbb R$ such that, almost everywhere on $(0,1)$, the potential $V(x)$ is a linear combination of the squared eigenfunctions and a constant:
			\begin{equation}\label{eq:V_via_Sr_boxed}
				V(x) = A + \frac{2}{C} S_{r_0+1}(x) = A + \sum_{k=1}^m b_k u_k^2(x) \quad \text{a.e. on } (0,1).
			\end{equation}
			\item There are $\varepsilon>0$ and a real-analytic function $\widetilde V\in C^\omega(-\varepsilon,1+\varepsilon)$ that agrees with $V(x)$ almost everywhere on $(0,1)$.
			\item If $V$ and the eigenfunctions $u_k$ satisfy periodic boundary conditions on the circle $\mathbb R/\mathbb Z$ of length $1$, then $\widetilde V$ is a real-analytic periodic function on the entire circle $\mathbb R/\mathbb Z$.
		\end{enumerate}
	\end{theorem}
	
	\begin{proof}
		1. Delete the terms with zero coefficients in \eqref{eq:uk_sq_dep} and renumber the remaining functions, retaining the notation $m$ for their number. We may thus assume that $c_k\ne0$ for all $k=1,\dots,m$ and that $m\ge2$, since for $m=1$ the equality $c_1u_1^2\equiv0$ would contradict $u_1\not\equiv0$.
		
		Consider the functions $S_r(x)$ and $T_r(x)$ defined in \eqref{eq:Sr_Tr_def}. Condition \eqref{eq:uk_sq_dep} gives $S_0(x)\equiv0$ on $[0,1]$. We claim that $S_0(x),S_1(x),\dots,S_{m-1}(x)$ cannot all vanish identically on $[0,1]$. Indeed, if
		\[
		S_r(x) = \sum_{k=1}^m \lambda_k^r \big(c_k u_k^2(x)\big) = 0, \qquad r = 0, 1, \dots, m-1, \quad \forall x \in [0,1],
		\]
		then invertibility of the Vandermonde matrix $(\lambda_k^r)_{r,k=0}^{m-1}$ for pairwise distinct $\lambda_k$ would imply $c_ku_k^2(x)=0$ for all $x\in[0,1]$ and all $k=1,\dots,m$. Since $c_k\ne0$, this would give $u_k(x)\equiv0$, contradicting the nontriviality of the solutions.
		
		Therefore there is a smallest index $r_0\in\{1,\dots,m-1\}$ for which $S_{r_0}(x)\not\equiv0$. Since $S_{r_0-1}(x)\equiv0$, Lemma~\ref{lem:Sr_Tr_identities}, specifically \eqref{eq:Sr_plus_1_const} with $r=r_0-1$, shows that $S_{r_0}(x)$ is identically equal to a nonzero constant:
		\begin{equation}\label{eq:Sp_const_nonzero}
			S_{r_0}(x) \equiv C \neq 0, \qquad \forall x \in [0,1].
		\end{equation}
		
		2. Since $S_{r_0}(x)\equiv C$, we have $S_{r_0}'(x)\equiv0$ and $S_{r_0}''(x)\equiv0$. Substituting $S_{r_0}=C$ and $S_{r_0}'=S_{r_0}''=0$ into \eqref{eq:star_identities} with $r=r_0$, we obtain almost everywhere on $(0,1)$
		\begin{equation}\label{eq:Tp_Sp1_rel}
			T_{r_0} = S_{r_0+1} - C V, \qquad T_{r_0}' = -S_{r_0+1}'.
		\end{equation}
		The second equality in \eqref{eq:Tp_Sp1_rel} means that $T_{r_0}+S_{r_0+1}\in W^{1,p}(0,1)$ has zero weak derivative. It is therefore identically equal to a real constant $B\in\mathbb R$ on $[0,1]$:
		\[
		T_{r_0}(x) + S_{r_0+1}(x) = B, \qquad \forall x \in [0,1].
		\]
		Eliminating $T_{r_0}=B-S_{r_0+1}$ from the first equality in \eqref{eq:Tp_Sp1_rel}, we obtain $B-S_{r_0+1}=S_{r_0+1}-CV$, and hence the representation \eqref{eq:V_via_Sr_boxed} almost everywhere on $(0,1)$:
		\begin{equation}\label{eq:V_explicit_rep}
			V(x) = A + \frac{2}{C} S_{r_0+1}(x) = A + \sum_{k=1}^m b_k u_k^2(x), \qquad \text{where } A := -\frac{B}{C}, \;\; b_k := \frac{2 c_k \lambda_k^{r_0+1}}{C}.
		\end{equation}
		
		3. The right-hand side of \eqref{eq:V_explicit_rep} belongs to $W^{2,p}(0,1)\hookrightarrow C^1[0,1]$. Substituting this continuous representation of $V(x)$ into \eqref{eq:sl_gen_uk}, we obtain a closed autonomous system of ordinary differential equations:
		\begin{equation}\label{eq:closed_poly_sys}
			u_k''(x) = \left( A + \sum_{j=1}^m b_j u_j^2(x) - \lambda_k \right) u_k(x), \qquad k = 1, \dots, m.
		\end{equation}
		Because the right-hand side of \eqref{eq:closed_poly_sys} is continuous on $[0,1]$ and $u_k'$ is absolutely continuous, integration shows that $u_k\in C^2[0,1]$ and that \eqref{eq:closed_poly_sys} holds classically at every $x\in[0,1]$.
		
		Rewrite the second-order system \eqref{eq:closed_poly_sys} as a first-order system for $Z(x)=(u_1(x),\dots,u_m(x),v_1(x),\dots,v_m(x))^T\in\mathbb R^{2m}$, where $v_k=u_k'$:
		\begin{equation}\label{eq:first_order_poly}
			Z'(x) = \mathcal{P}(Z(x)),
		\end{equation} where \[ \mathcal{P}(u, v) := \left( v_1, \dots, v_m, \, \Big(A + \sum_{j=1}^m b_j u_j^2 - \lambda_1\Big)u_1, \dots, \Big(A + \sum_{j=1}^m b_j u_j^2 - \lambda_m\Big)u_m \right)^T.\]
		
		The map $\mathcal P:\mathbb R^{2m}\to\mathbb R^{2m}$ is polynomial with real coefficients and hence real analytic on all of $\mathbb R^{2m}$. By the analyticity theorem for ODEs with analytic right-hand sides, in a neighborhood of every $x_0\in[0,1]$---including the endpoints $x_0=0$ and $x_0=1$, where one solves the Cauchy problems with data $Z(0)$ and $Z(1)$---there is a unique local real-analytic solution. By uniqueness, it agrees with $Z(x)$ wherever both are defined. Consequently, the functions $u_k(x)$ extend analytically to real-analytic functions $\widetilde u_k(x)$ on some open interval $(-\varepsilon,1+\varepsilon)$, $\varepsilon>0$.
		
		Setting
		\[
		\tilde{V}(x) := A + \sum_{k=1}^m b_k \tilde{u}_k^2(x), \qquad x \in (-\varepsilon, 1 + \varepsilon),
		\]
		we obtain a real-analytic function $\widetilde V$ on $(-\varepsilon,1+\varepsilon)$ that agrees with $V(x)$ almost everywhere on $(0,1)$ by \eqref{eq:V_explicit_rep}. Finally, if $V$ and $u_k$ are $1$-periodic, so that $u_k(0)=u_k(1)$ and $u_k'(0)=u_k'(1)$, then $Z(0)=Z(1)$. Uniqueness for the Cauchy problem associated with the autonomous system \eqref{eq:first_order_poly} implies that the analytic solutions $\widetilde u_k(x)$ and the potential $\widetilde V(x)$ are $1$-periodic real-analytic functions on $\mathbb R$, equivalently on the circle $\mathbb R/\mathbb Z$. The proof is complete.
	\end{proof}
	
	\subsection{Existence and representation of OISP solutions in $L^p(0,1)$}
	
	To write the analytic representation of an optimal potential in $L^p(0,1)$, we use the standard notation for the signed power of a real number $s\in\mathbb R$ with exponent $r>0$:
	\begin{equation}\label{eq:signed_power}
		s^{\langle r \rangle} := |s|^{r-1} s = |s|^r \sgn(s) \quad (0^{\langle r \rangle} := 0).
	\end{equation}
	The map $s\mapsto s^{\langle r\rangle}$ is an odd strictly increasing homeomorphism of $\mathbb R$ onto itself, with inverse $t\mapsto t^{\langle1/r\rangle}$. In particular, conjugate exponents $p,q\in(1,\infty)$ satisfy $(p-1)(q-1)=1$, and therefore
	\begin{equation}\label{eq:signed_power_inv}
		\left( s^{\langle p-1 \rangle} \right)^{\langle q-1 \rangle} = s, \quad \forall s \in \mathbb{R}, \quad \text{where } q - 1 = \frac{1}{p-1}.
	\end{equation}
	
	\begin{theorem}[Existence and analytic representation of an OISP solution in $L^p(0,1)$]\label{thm:main_lp}
		Let $1<p<\infty$, $q=\frac{p}{p-1}$, and let a reference potential $V_0\in L^p(0,1)$, indices $1\le i_1<i_2<\dots<i_m$, and a strictly increasing vector of target eigenvalues $\Lambda^*=(\lambda_1^*,\dots,\lambda_m^*)^T\in\Omega_m$ be given. Then:
		\begin{enumerate}
			\item \textbf{(Existence of a global solution).} The optimization inverse spectral problem $\mathcal P_p(m)$ in \eqref{eq:oisp_dist}--\eqref{eq:oisp_var} has at least one global solution $\hat V\in\ISM_m^p(\Lambda^*)$.
			\item \textbf{(Representation through squared eigenfunctions).} For every local or global minimizer $\hat V\in\ISM_m^p(\Lambda^*)$ of \eqref{eq:oisp_var}, there exist unique Lagrange multipliers $\alpha_1,\dots,\alpha_m\in\mathbb R$ such that the optimal potential has the representation
			\begin{equation}\label{eq:lp_rep}
				\hat{V}(x) = V_0(x) - \left( \sum_{k=1}^m \alpha_k \varphi_{i_k}^2(\hat{V})(x) \right)^{\langle \frac{1}{p-1} \rangle} ,
			\end{equation}
			where $\varphi_{i_k}(\hat V)$ are the $L^2(0,1)$-normalized eigenfunctions of $\mathcal L[\hat V]$. If $\lambda_{i_k}(V_0)\ne\lambda_k^*$ for at least one $k\in\{1,\dots,m\}$, then $\sum_{k=1}^m|\alpha_k|>0$.
			\item \textbf{(Regularity of the correction).} The optimal perturbation $\hat V-V_0$ is continuous on $[0,1]$ and vanishes at $x=0$ and $x=1$: $\hat V-V_0\in C_0[0,1]$. Moreover, for $1<p\le2$ the correction belongs to $C^1[0,1]$.
			\item \textbf{(Relation to a nonlinear boundary value system).} Set $\sigma_k=\sgn(\alpha_k)\in\{-1,0,+1\}$ and $\hat u_k(x)=\sqrt{|\alpha_k|}\,\varphi_{i_k}(\hat V)(x)$, $k=1,\dots,m$. Then the vector function $(\hat u_1,\dots,\hat u_m)\in(W^{2,p}(0,1)\cap W_0^{1,2}(0,1))^m$ satisfies the boundary value problem for the nonlinear system
			\begin{equation}\label{eq:lp_nls_sys}
				\left\{ \begin{aligned}
					-\hat{u}_k''(x) + V_0(x)\hat{u}_k(x) &= \lambda_k^* \hat{u}_k(x) + \left( \sum_{j=1}^m \sigma_j \hat{u}_j^2(x) \right)^{\langle \frac{1}{p-1} \rangle} \hat{u}_k(x), \quad x \in (0,1), \\
					\hat{u}_k(0) = \hat{u}_k(1) &= 0, \quad k = 1, \dots, m,
				\end{aligned} \right.
			\end{equation}
			The functions $\{\hat u_k\}_{k=1}^m$ are pairwise orthogonal in $L^2(0,1)$, and, if $\alpha_k\ne0$, then $\hat u_k$ has exactly $i_k-1$ simple interior zeros. The vector $\hat u$ is nontrivial if $F_m(V_0)\ne\Lambda^*$. If $F_m(V_0)=\Lambda^*$, then the unique minimum is $V_0$, and all multipliers and all scaled components vanish.
		\end{enumerate}
	\end{theorem}
	
	\begin{proof}
		\textit{Step 1.}
		The general theory of inverse spectral problems \cite{Poschel1987}, as well as the more general direct variational construction of potential wells separated by barriers proved in Theorem~3.1 of \cite{Valeev2026Graph}, shows that for every strictly increasing set $\lambda_1^*<\dots<\lambda_m^*$ there is a bounded potential $V_{\mathrm{adm}}\in L^\infty(0,1)$ such that $\lambda_{i_k}(V_{\mathrm{adm}})=\lambda_k^*$, $k=1,\dots,m$. Since $L^\infty(0,1)\subset L^p(0,1)$ on the finite interval $(0,1)$, the set $\ISM_m^p(\Lambda^*)$ is nonempty and the infimum in \eqref{eq:oisp_var} is finite:
		\begin{equation}\label{eq:beta_p_fin}
			0 \le \beta_p := \inf_{V \in \ISM_m^p(\Lambda^*)} J_p(V) \le J_p(V_{\mathrm{adm}}) < +\infty.
		\end{equation}

		Choose a minimizing sequence $\{V_j\}_{j=1}^\infty\subset\ISM_m^p(\Lambda^*)$ such that $J_p(V_j)\to\beta_p$ and $\lambda_{i_k}(V_j)=\lambda_k^*$ for all $k=1,\dots,m$ and $j\ge1$. The triangle inequality gives the coercivity estimate
		\[
		\|V_j\|_{L^p} \le \|V_0\|_{L^p} + \|V_j - V_0\|_{L^p} = \|V_0\|_{L^p} + \left( p J_p(V_j) \right)^{1/p} \le \|V_0\|_{L^p} + \left( p(\beta_p + 1) \right)^{1/p}
		\]
		for all sufficiently large $j$. Thus $\{V_j\}$ is bounded in $L^p(0,1)$.
		
		\textit{Step 2 (existence of a minimizer).}
		For $1<p<\infty$, the Banach space $L^p(0,1)$ is reflexive. Hence a subsequence $\{V_{j_s}\}_{s=1}^\infty$ converges weakly in $L^p(0,1)$ to some $\hat V\in L^p(0,1)$: $V_{j_s}\rightharpoonup\hat V$ as $s\to\infty$. By Lemma~\ref{lem:weak_cont_lp}, this weak convergence implies convergence of every eigenvalue:
		\begin{equation}\label{eq:lim_lambda_hat}
			\lambda_{i_k}(\hat{V}) = \lim_{s \to \infty} \lambda_{i_k}(V_{j_s}) = \lambda_k^*, \quad k = 1, \dots, m.
		\end{equation}
		Thus $\hat V\in\ISM_m^p(\Lambda^*)$. On the other hand, $J_p(V)=\frac1p\|V-V_0\|_{L^p}^p$ is weakly sequentially lower semicontinuous, so $(V_{j_s}-V_0)\rightharpoonup\hat V-V_0$ gives
		\begin{equation}\label{eq:wlsc_jp}
			\beta_p \le J_p(\hat{V}) \le \liminf_{s \to \infty} J_p(V_{j_s}) = \beta_p.
		\end{equation}
		Hence $J_p(\hat V)=\beta_p$, so $\hat V$ is a global solution of the OISP $\mathcal P_p(m)$. Since $L^p(0,1)$ is uniformly convex, the Radon--Riesz theorem implies that $V_{j_s}\to\hat V$ strongly in $L^p(0,1)$.
		
		\textit{Step 3 (Lagrange multiplier rule).}
		Let $\hat V\in\ISM_m^p(\Lambda^*)$ be any local or global minimizer of \eqref{eq:oisp_var}. For $1<p<\infty$, the functional $J_p:L^p(0,1)\to\mathbb R$ is continuously Fr\'echet differentiable on all of $L^p(0,1)$, and its Fr\'echet derivative at $\hat V$ is given by the canonical duality map from $L^p(0,1)$ to $L^q(0,1)$:
		\begin{equation}\label{eq:djp_dual}
			\begin{aligned}
 \langle DJ_p(\hat{V}),h\rangle
 &=\int_0^1(\hat V-V_0)^{\langle p-1\rangle}h\,dx\\
 &=\int_0^1|\hat V-V_0|^{p-2}(\hat V-V_0)h\,dx.
 \end{aligned}
		\end{equation}
		for every $h\in L^p(0,1)$. Indeed, since $\hat V-V_0\in L^p(0,1)$, the function $w(x):=(\hat V(x)-V_0(x))^{\langle p-1\rangle}$ satisfies $|w(x)|^q=|\hat V(x)-V_0(x)|^{(p-1)q}=|\hat V(x)-V_0(x)|^p\in L^1(0,1)$. Thus $w\in L^q(0,1)=(L^p(0,1))^*$.
		
		By Lemma~\ref{lem:frechet_indep_lp}, the constraint functionals $\lambda_{i_k}:L^p(0,1)\to\mathbb R$, $k=1,\dots,m$, belong to $C^1(L^p(0,1))$, and their Fr\'echet derivatives $D\lambda_{i_k}(\hat V)=\varphi_{i_k}^2(\hat V)\in L^q(0,1)$ are linearly independent. Hence $\hat V$ is a regular point of the constraint map $F_m$. The Ljusternik multiplier theorem in Banach spaces yields unique real Lagrange multipliers $\alpha_1,\dots,\alpha_m$ such that
		\begin{equation}\label{eq:lagrange_lusternik}
			DJ_p(\hat{V}) + \sum_{k=1}^m \alpha_k D\lambda_{i_k}(\hat{V}) = 0 \quad \text{in } L^q(0,1).
		\end{equation}
		In integral form, for every test function $h\in L^p(0,1)$ this means
		\begin{equation}\label{eq:lagrange_int}
			\int_0^1 \left[ \big(\hat{V}(x) - V_0(x)\big)^{\langle p-1 \rangle} + \sum_{k=1}^m \alpha_k \varphi_{i_k}^2(\hat{V})(x) \right] h(x) \, dx = 0,
		\end{equation}
		and the fundamental lemma of the calculus of variations gives, almost everywhere on $(0,1)$,
		\begin{equation}\label{eq:lagrange_ae}
			\big(\hat{V}(x) - V_0(x)\big)^{\langle p-1 \rangle} = -\sum_{k=1}^m \alpha_k \varphi_{i_k}^2(\hat{V})(x).
		\end{equation}
		
		\textit{Step 4 (derivation of \eqref{eq:lp_rep}).}
		Apply to both sides of \eqref{eq:lagrange_ae} the inverse odd power map $t\mapsto t^{\langle q-1\rangle}=t^{\langle\frac1{p-1}\rangle}$, using \eqref{eq:signed_power_inv}. Almost everywhere on $(0,1)$, this gives
		\begin{equation}\label{eq:lp_rep_inverted}
			\hat{V}(x) - V_0(x) = -\left( \sum_{k=1}^m \alpha_k \varphi_{i_k}^2(\hat{V})(x) \right)^{\langle \frac{1}{p-1} \rangle},
		\end{equation}
		which is precisely \eqref{eq:lp_rep}. Uniqueness of $\alpha_1,\dots,\alpha_m$ follows from uniqueness of the left-hand side of \eqref{eq:lagrange_ae} and linear independence of $\varphi_{i_1}^2(\hat V),\dots,\varphi_{i_m}^2(\hat V)$ \cite{IlyasovValeev2021}. If $\lambda_{i_k}(V_0)\ne\lambda_k^*$ for at least one $k$, then $\hat V\ne V_0$ in $L^p(0,1)$, and \eqref{eq:lp_rep} implies $\sum_{k=1}^m|\alpha_k|>0$.
		
		\textit{Step 5 (the nonlinear Schr\"odinger system \eqref{eq:lp_nls_sys}).}
		Since $\hat V\in L^p(0,1)$, every eigenfunction $\varphi_{i_k}(\hat V)$ belongs to $W^{2,p}(0,1)\hookrightarrow C^1[0,1]$ and vanishes at $x=0$ and $x=1$. Thus $\Phi(x):=\sum_{k=1}^m\alpha_k\varphi_{i_k}^2(\hat V)(x)$ belongs to $C^1[0,1]$ and satisfies $\Phi(0)=\Phi(1)=0$. The map $s\mapsto s^{\langle\frac1{p-1}\rangle}$ is continuous on $\mathbb R$ for every $p\in(1,\infty)$ and continuously differentiable on $\mathbb R$ when $\frac1{p-1}\ge1$, that is, when $1<p\le2$. Therefore the optimal correction $\hat V(x)-V_0(x)=-(\Phi(x))^{\langle\frac1{p-1}\rangle}$ belongs to $C_0[0,1]$ for all $1<p<\infty$ and to $C^1[0,1]$ for $1<p\le2$.
		
		Finally, set $\sigma_k=\sgn(\alpha_k)$ and $\hat u_k(x)=\sqrt{|\alpha_k|}\,\varphi_{i_k}(\hat V)(x)$, so that $\alpha_k\varphi_{i_k}^2=\sigma_k\hat u_k^2$ \cite{IlyasovValeev2021}. Then \eqref{eq:lp_rep} becomes
		\begin{equation}\label{eq:lp_pot_u}
			\hat{V}(x) = V_0(x) - \left( \sum_{j=1}^m \sigma_j \hat{u}_j^2(x) \right)^{\langle \frac{1}{p-1} \rangle}.
		\end{equation}
		Substituting \eqref{eq:lp_pot_u} for $\hat V(x)$ into $-\varphi_{i_k}''(\hat V)+\hat V\varphi_{i_k}(\hat V)=\lambda_k^*\varphi_{i_k}(\hat V)$ and multiplying by $\sqrt{|\alpha_k|}$ yields the closed nonlinear system \eqref{eq:lp_nls_sys}. Orthogonality of the functions $\hat u_k$ in $L^2(0,1)$ and, when $\alpha_k\ne0$, the fact that $\hat u_k$ has exactly $i_k-1$ interior zeros in $(0,1)$ follow from self-adjointness of $\mathcal L[\hat V]$ and Sturm's oscillation theorem.
	\end{proof}
	
	\begin{remark}[Properties of the optimal potential for different values of $p$]\label{rem:p_regimes}
Formula~\eqref{eq:lp_rep} describes the action of the signed power on $\Phi=\sum_k\alpha_k\varphi_{i_k}^2$:
\begin{itemize}
\item For $p=2$, the exponent is one, and we recover the classical representation~\eqref{eq:p2_rep} and the cubic system~\eqref{eq:p2_nls}~\cite{ValeevIlyasov2018MZ,IlyasovValeev2021}.
\item For $1<p<2$, the exponent is greater than one, and the correction belongs to $C^1[0,1]$.
\item For $p>2$, the exponent is less than one; the correction remains continuous, but differentiability at zeros of $\Phi$ does not follow from this formula.
\end{itemize}
\end{remark}
\begin{remark}
The limiting problems in $L^1$ and $L^\infty$ require separate formulations and proofs.

We record only the following \textbf{two conjectures}:

1) As $p\to1+$, optimal corrections tend to linear combinations of Dirac delta measures concentrated at extremal points.

2) As $p\to+\infty$, the correction $\hat V(x)-V_0(x)$ tends to a piecewise constant function of the form $-C\sgn\!\left(\sum_{k=1}^m\alpha_k\varphi_{i_k}^2(x)\right)$ of bang--bang type that minimizes the uniform norm $\|\hat V-V_0\|_{L^\infty(0,1)}$.
\end{remark}
	
	\section{Isolation and generic uniqueness of OISP solutions}
	
	Because the isospectral manifold $\ISM_m^p(\Lambda^*)$ is a nonlinear and nonconvex subset of $L^p(0,1)$, strict convexity of the objective functional $J_p(V)=\frac1p\|V-V_0\|_{L^p}^p$ alone does not guarantee global uniqueness of a nearest point for an arbitrary reference potential $V_0$. Nevertheless, OISP solutions possess local isolation, compactness, and generic uniqueness properties.
	
	\subsection{Isolation of solutions and finiteness of the set of global minima}
	
	Let $\hat V\in\ISM_m^p(\Lambda^*)$ be an OISP solution. By Theorem~\ref{thm:main_lp}, $\hat V$ satisfies the Euler--Lagrange equation \eqref{eq:lp_rep}:
	\begin{equation}\label{eq:euler_lagrange_op}
		\hat{V} = V_0 - \mathcal{N}_p(\alpha, \hat{V}), \quad \text{where } \mathcal{N}_p(\alpha, V)(x) := \left( \sum_{k=1}^m \alpha_k \varphi_{i_k}^2(V)(x) \right)^{\langle \frac{1}{p-1} \rangle}.
	\end{equation}
	
	\begin{lemma}[Compactness of the set of global OISP solutions]\label{lem:compact_sol_set}
		For every fixed $V_0\in L^p(0,1)$, $1<p<\infty$, and $\Lambda^*\in\Omega_m$, the set $\mathcal S_{\min}(V_0;\Lambda^*)$ of all global solutions of $\mathcal P_p(m)$ in \eqref{eq:oisp_var} is a nonempty compact subset of the Banach space $L^p(0,1)$. Moreover, it is compact in the affine space $V_0+C_0[0,1]$.
	\end{lemma}
	
	\begin{proof}
		Nonemptiness of $\mathcal S_{\min}(V_0;\Lambda^*)$ was proved in Theorem~\ref{thm:main_lp}. Let $\{\hat V_j\}_{j=1}^\infty\subset\mathcal S_{\min}(V_0;\Lambda^*)$ be any sequence of global minimizers. Then $\|\hat V_j-V_0\|_{L^p}=(p\beta_p)^{1/p}$ for all $j\ge1$, so $\{\hat V_j\}$ is bounded in $L^p(0,1)$. Passing to a subsequence, $\hat V_{j_s}\rightharpoonup\hat V_*$ weakly in $L^p(0,1)$. By Lemma~\ref{lem:weak_cont_lp}, $\hat V_*\in\ISM_m^p(\Lambda^*)$, and hence $\|\hat V_*-V_0\|_{L^p}\ge(p\beta_p)^{1/p}$. Weak lower semicontinuity of the norm gives the reverse inequality, $\|\hat V_*-V_0\|_{L^p}\le\lim_{s\to\infty}\|\hat V_{j_s}-V_0\|_{L^p}=(p\beta_p)^{1/p}$. Thus $\|\hat V_*-V_0\|_{L^p}=\lim_{s\to\infty}\|\hat V_{j_s}-V_0\|_{L^p}$. Uniform convexity of $L^p(0,1)$ for $1<p<\infty$, through the Radon--Riesz theorem, now yields strong convergence $\hat V_{j_s}\to\hat V_*$ in $L^p(0,1)$ and $\hat V_*\in\mathcal S_{\min}(V_0;\Lambda^*)$. For the last assertion, set $g_k(V)=\varphi_{i_k}(V)^2$ and $G(V)_{kl}=\int g_k(V)g_l(V)$. Formula~\eqref{eq:lagrange_ae} gives
\[
 \alpha^{(j_s)}=-G(\hat V_{j_s})^{-1}
 \left(\int (\hat V_{j_s}-V_0)^{\langle p-1\rangle}g_l(\hat V_{j_s})\right)_{l=1}^m.
\]
The duality map $L^p\to L^q$ is continuous, the eigenfunctions converge in $C^1$, and the limiting Gram matrix is invertible. Consequently, the multipliers converge. Formula~\eqref{eq:lp_rep} and uniform continuity of the signed power on bounded intervals imply uniform convergence of the corrections in $C_0[0,1]$.
	\end{proof}
	
	For the second-order conditions in this subsection, set $p=2$ and write $M=\ISM_m^2(\Lambda^*)$ and $H=L^2(0,1)$. The Hessian statements below are not used for $p\ne2$. The Lagrangian is
\begin{equation}\label{eq:lagrangian_def}
 \mathcal L_2(V,\alpha)=\tfrac12\|V-V_0\|_2^2+
 \sum_{k=1}^m\alpha_k(\lambda_{i_k}(V)-\lambda_k^*).
\end{equation}
Let $P_kf=\langle f,\varphi_{i_k}\rangle\varphi_{i_k}$ and
$R_k=(\mathcal L[V]-\lambda_{i_k}(V))^{-1}_{\perp}(I-P_k)$; thus $R_k$ vanishes on the eigenfunction. Then, for $h\in H$,
\begin{equation}\label{eq:second_var_lambda}
 D^2\lambda_{i_k}(V)[h,h]=-2\langle R_k(h\varphi_{i_k}),h\varphi_{i_k}\rangle_2.
\end{equation}
Indeed, differentiating the eigenvalue equation and the normalization gives $D\varphi_{i_k}(V)[h]=-R_k(h\varphi_{i_k})$; one then differentiates the first-variation formula. The operator $K_k=-2M_{\varphi_{i_k}}R_kM_{\varphi_{i_k}}$, where $M_f$ denotes multiplication by $f$, is compact and self-adjoint: $R_k$ is compact and the eigenfunction is bounded. These operators depend continuously on $V$ in the operator norm; this also follows from analyticity of $V\mapsto\lambda_{i_k}(V)$.
\begin{equation}\label{eq:hessian_p2}
 H_{V,\alpha}:=D^2_{VV}\mathcal L_2(V,\alpha)=I+\sum_k\alpha_kK_k.
\end{equation}

\begin{definition}[Nondegeneracy of a stationary point and of a minimum]\label{def:nondeg}
Set
\begin{equation}\label{eq:tangent_space}
 T_V M=\{h\in H:\ \langle h,\varphi_{i_k}(V)^2\rangle=0,\ k=1,\ldots,m\}.
\end{equation}
A stationary pair $(V,\alpha)$ is called \emph{Morse nondegenerate} if
$A_{V,\alpha}=P_{T_VM}H_{V,\alpha}|_{T_VM}$ is invertible on $T_VM$.
A minimum is called \emph{strongly nondegenerate} if, for some $\gamma>0$,
\begin{equation}\label{eq:hessian_coercive}
 \langle H_{V,\alpha}h,h\rangle\ge\gamma\|h\|_2^2,
 \qquad h\in T_VM.
\end{equation}
For a local minimum, these two conditions are equivalent. Indeed, the second-order necessary condition gives $A_{V,\alpha}\ge0$, while $A_{V,\alpha}=I_{T_VM}+\text{a compact self-adjoint operator}$. Invertibility then provides a positive spectral gap. For an arbitrary stationary point, invertibility does not imply positivity.
\end{definition}

\begin{theorem}[Isolation and generic nondegeneracy for $p=2$]\label{thm:isolated_finite}
For fixed indices and $\Lambda^*\in\Omega_m$, the following statements hold:
\begin{enumerate}
\item Every Morse nondegenerate stationary point is isolated among stationary points. If it is a local minimum, then there exist $c,\delta>0$ such that
$J_2(W)\ge J_2(V)+c\|W-V\|_2^2$ for $W\in M$ with $\|W-V\|_2<\delta$.
\item If all global minima are nondegenerate, then their number is finite.
\item The set of centers $V_0\in H$ for which all stationary points are Morse nondegenerate is residual. The set of centers for which all global minima are nondegenerate is open and dense in $H$.
\end{enumerate}
\end{theorem}
\begin{proof}
1. Locally, $M$ is the graph $V+h+\Psi(h)$ over $T_VM$, where $\Psi(0)=D\Psi(0)=0$. The derivative at zero of the gradient of the restriction of $J_2$ equals $A_{V,\alpha}$. The inverse function theorem shows that the stationary point is isolated. At a minimum, Definition~\ref{def:nondeg} and Taylor's formula give
$J_2(V+h+\Psi(h))\ge J_2(V)+\gamma\|h\|_2^2/4$ for sufficiently small $h$. Since $\|h+\Psi(h)\|_2\le2\|h\|_2$ in a sufficiently small neighborhood, the required estimate in the ambient norm follows.

2. The set of global minima is compact by Lemma~\ref{lem:compact_sol_set}. If there were infinitely many isolated minima, they would have a limit point in this set, contradicting part~1.

3. Consider the analytic map
\begin{equation}\label{eq:normal_bundle_map}
 \mathcal E:M\times\mathbb R^m\to H,\qquad
 \mathcal E(V,\alpha)=V+B_V\alpha,
 \quad B_V\alpha=\sum_k\alpha_k\varphi_{i_k}(V)^2.
\end{equation}
The stationary pairs for the center $V_0$ form $\mathcal E^{-1}(V_0)$. For $h\in T_VM$ and $\beta\in\mathbb R^m$,
\[
 D\mathcal E(V,\alpha)[h,\beta]=H_{V,\alpha}h+B_V\beta.
\]
The operator $(h,\beta)\mapsto h+B_V\beta$ is an isomorphism from $T_VM\times\mathbb R^m$ onto $H$, and its difference from $D\mathcal E$ is compact. Thus $\mathcal E$ is a Fredholm map of index zero. Relative to the decomposition $H=T_VM\oplus(T_VM)^\perp$, its derivative has triangular form with diagonal blocks $A_{V,\alpha}$ and $B_V:\mathbb R^m\to(T_VM)^\perp$. Hence regularity of the pair is equivalent to Morse nondegeneracy.
As a separable Hilbert manifold, $M\times\mathbb R^m$ has a countable base. By the Sard--Smale theorem~\cite[Theorem~(1.3)]{Smale1965}, the regular values of $\mathcal E$ form a residual subset of $H$. No global openness of this set is asserted here.

We separately prove that the set $\mathcal O$ of centers with nondegenerate global minima is open. Let $X_n\to X$ in $H$ and $V_n\in\Pi_M(X_n)$. The Lipschitz property of the distance function gives
\[
 \big|\|V_n-X\|_2-d_M(X)\big|\le2\|X_n-X\|_2\to0.
\]
Boundedness, weak closedness of $M$, and the Radon--Riesz property allow us to extract a subsequence such that $V_n\to V$ in norm, where $V\in\Pi_M(X)$. The multipliers converge by
\begin{equation}\label{eq:multiplier_gram_moving}
 \alpha_n=-G(V_n)^{-1}
 \big(\langle V_n-X_n,\varphi_{i_j}(V_n)^2\rangle\big)_{j=1}^m.
\end{equation}
Here $G$ is the Gram matrix of the squared eigenfunctions. In local coordinates for the tangent spaces, $A_{V_n,\alpha_n}\to A_{V,\alpha}$. If $X\in\mathcal O$, the limiting operator is invertible, and so are all sufficiently close operators. Therefore a sequence of centers outside $\mathcal O$ cannot converge to $X\in\mathcal O$, proving openness. Density of $\mathcal O$ follows from residuality of the regular values of $\mathcal E$ and the Baire theorem.
\end{proof}

\begin{remark}
Restricting the second-order theory to $p=2$ is essential for this proof. For $1<p<2$, the functional $J_p$ is not $C^2$ on all of $L^p$. For $p>2$, its Hessian at a stationary point has the form $(p-1)\int|V-V_0|^{p-2}h^2$ with a bounded weight and is not coercive with respect to $\|h\|_p^2$. Functions of unit $L^p$-norm concentrated on small intervals have $L^2$-norm tending to zero; subtracting their finite-dimensional normal component preserves this property while enforcing the tangent constraints. Isolation conditions for $p\ne2$ require a separate analysis.
\end{remark}

	\subsection{Uniqueness and generic uniqueness of OISP solutions}
	
	\begin{theorem}[Global uniqueness for the first eigenvalue when $\lambda_1^*\ge\lambda_1(V_0)$]\label{thm:uniq_lambda1}
		Let $1<p<\infty$, $m=1$, and $i_1=1$. Let $V_0\in L^p(0,1)$ be an arbitrary reference potential and let the target first eigenvalue $\lambda_1^*\in\mathbb R$ satisfy $\lambda_1^*\ge\lambda_1(V_0)$. Then the optimization inverse spectral problem $\mathcal P_p(1)$
		\begin{equation}\label{eq:oisp_m1}
			\frac{1}{p}\|V - V_0\|_{L^p(0,1)}^p \longrightarrow \min, \quad \lambda_1(V) = \lambda_1^*, \quad V \in L^p(0,1),
		\end{equation}
		has a \textbf{unique global solution} $\hat V\in L^p(0,1)$. If $\lambda_1^*>\lambda_1(V_0)$, then this solution has the form
		\begin{equation}\label{eq:oisp_m1_rep}
			\hat{V}(x) = V_0(x) + \big( |\alpha_1| \varphi_1^2(\hat{V})(x) \big)^{\frac{1}{p-1}} = V_0(x) + \hat{u}_1^{\frac{2}{p-1}}(x),
		\end{equation}
		where $\hat u_1(x)>0$ on $(0,1)$ is the unique positive solution of the defocusing nonlinear boundary value problem
		\begin{equation}\label{eq:oisp_m1_nls}
			-\hat{u}_1''(x) + \left( V_0(x) + \hat{u}_1^{\frac{2}{p-1}}(x) \right) \hat{u}_1(x) = \lambda_1^* \hat{u}_1(x), \quad \hat{u}_1(0) = \hat{u}_1(1) = 0.
		\end{equation}
	\end{theorem}
	
	\begin{proof}
		1. Introduce the superlevel set of the first eigenvalue in $L^p(0,1)$ \cite{Sadovnichii2022DE,Sadovnichii2026R3}:
		\begin{equation}\label{eq:superlevel_set}
			\mathcal{K}_{\lambda_1^*}^p := \left\{ V \in L^p(0,1) : \lambda_1(V) \ge \lambda_1^* \right\}.
		\end{equation}
		By the Rayleigh--Ritz variational principle, the first eigenvalue is the infimum of a family of functionals affine in $V$:
		\begin{equation}\label{eq:rayleigh_concave}
			\lambda_1(V) = \inf_{\substack{u \in W_0^{1,2}(0,1) \\ \|u\|_{L^2}=1}} \left( \int_0^1 |u'(x)|^2 \, dx + \int_0^1 V(x)|u(x)|^2 \, dx \right).
		\end{equation}
		Since the infimum of any family of affine functionals is concave, the map $V\mapsto\lambda_1(V)$ is concave on $L^p(0,1)$. Thus, for all $V_1,V_2\in L^p(0,1)$ and $\theta\in[0,1]$,
		\[
		\lambda_1(\theta V_1 + (1-\theta)V_2) \ge \theta \lambda_1(V_1) + (1-\theta)\lambda_1(V_2).
		\]
		Concavity and continuity of $\lambda_1(\cdot)$ imply that $\mathcal K_{\lambda_1^*}^p$ is a \textit{closed convex} subset of $L^p(0,1)$ \cite{Sadovnichii2026R3}.
		
		2. For $1<p<\infty$, the Banach space $L^p(0,1)$ is reflexive and strictly convex. Therefore the nonempty closed convex set $\mathcal K_{\lambda_1^*}^p$ contains a unique element $\hat V$ minimizing $\|V-V_0\|_{L^p}$ over $\mathcal K_{\lambda_1^*}^p$.
		
		3. We show that, when $\lambda_1^*\ge\lambda_1(V_0)$, this unique element $\hat V$ lies on the boundary $\partial\mathcal K_{\lambda_1^*}^p=\ISM_1^p(\lambda_1^*)$, that is, $\lambda_1(\hat V)=\lambda_1^*$. If $\lambda_1(\hat V)>\lambda_1^*\ge\lambda_1(V_0)$, continuity of $t\mapsto\lambda_1((1-t)V_0+t\hat V)$ on $[0,1]$ would give $t_0\in[0,1)$ such that $\lambda_1((1-t_0)V_0+t_0\hat V)=\lambda_1^*$. Then $V_{t_0}:=(1-t_0)V_0+t_0\hat V$ would belong to $\ISM_1^p(\lambda_1^*)\subset\mathcal K_{\lambda_1^*}^p$ but would be strictly closer to $V_0$:
		\[
		\|V_{t_0} - V_0\|_{L^p} = t_0 \|\hat{V} - V_0\|_{L^p} < \|\hat{V} - V_0\|_{L^p},
		\]
		contradicting the minimality of $\hat V$. Since $\ISM_1^p(\lambda_1^*)\subset\mathcal K_{\lambda_1^*}^p$ and $\hat V\in\ISM_1^p(\lambda_1^*)$, the point $\hat V$ is the unique global solution of the OISP \eqref{eq:oisp_m1}.
		
		4. If $\lambda_1^*>\lambda_1(V_0)$, representation \eqref{eq:lp_rep} gives $\hat V(x)=V_0(x)-\sgn(\alpha_1)(|\alpha_1|\varphi_1^2(\hat V)(x))^{\frac1{p-1}}$. Since $\varphi_1(\hat V)(x)>0$ on $(0,1)$, the case $\alpha_1\ge0$ would imply $\hat V(x)\le V_0(x)$ and, by monotonicity of the first eigenvalue, $\lambda_1(\hat V)\le\lambda_1(V_0)<\lambda_1^*$, a contradiction. Hence $\alpha_1<0$, so $\sigma_1=-1$, and \eqref{eq:oisp_m1_rep}--\eqref{eq:oisp_m1_nls} follow.
	
        5. We also prove uniqueness among all positive solutions of \eqref{eq:oisp_m1_nls}. If $u>0$ is such a solution and $V=V_0+u^{2/(p-1)}$, then $u$ is the first eigenfunction of the operator with potential $V$, so $\lambda_1(V)=\lambda_1^*$. Moreover, $DJ_p(V)=u^2=\|u\|_2^2\varphi_1(V)^2$. For $W\in\mathcal K_{\lambda_1^*}^p$, concavity of the first eigenvalue gives
\[
 0\le\lambda_1(W)-\lambda_1^*\le\int\varphi_1(V)^2(W-V).
\]
Consequently, $DJ_p(V)[W-V]\ge0$. Convexity of $J_p$ makes $V$ a minimizer on $\mathcal K_{\lambda_1^*}^p$, and strict convexity makes it unique. Then $u=(V-V_0)^{(p-1)/2}$ is uniquely determined.
        \end{proof}
	
	\begin{theorem}[Uniqueness near the isospectral manifold for $p=2$]\label{thm:tubular_uniq}
Let $\Lambda^*\in\Omega_m$. There is an open neighborhood $\mathcal U\supset M=\ISM_m^2(\Lambda^*)$ in $L^2(0,1)$ such that, for every center $V_0\in\mathcal U$, the OISP has a unique global solution depending smoothly on $V_0$.
\end{theorem}
\begin{proof}
At $(V_*,0)$, the derivative of the map in \eqref{eq:normal_bundle_map} is $(h,\beta)\mapsto h+B_{V_*}\beta$, which is an isomorphism. The inverse function theorem gives a unique local smooth stationary branch for centers near $V_*$. We show that every global minimum lies on this branch. For any $\widehat V\in\Pi_M(V_0)$, admissibility of $V_*$ gives
\[
 \|\widehat V-V_*\|_2\le\|\widehat V-V_0\|_2+\|V_0-V_*\|_2
 \le2\|V_0-V_*\|_2.
\]
Hence, as $V_0\to V_*$, all global minima converge to $V_*$. Formula~\eqref{eq:multiplier_gram_moving} shows that their multipliers tend to zero. Thus, for every sufficiently close center, each minimizing pair lies in the same invertibility neighborhood of $\mathcal E$. Existence of a minimum has already been established, so it is unique and coincides with the smooth local branch. Taking the union of these open neighborhoods over $V_*\in M$ gives $\mathcal U$.
\end{proof}

\begin{theorem}[Generic global uniqueness in the sense of Baire category]\label{thm:generic_uniq_baire}
		Let $1<p<\infty$, let $m\ge1$ be arbitrary, let $1\le i_1<\dots<i_m$, and let $\Lambda^*\in\Omega_m$. Denote by $\mathcal U_{\mathrm{uniq}}(\Lambda^*)$ the set of all potentials $V_0\in L^p(0,1)$ for which the global solution of the OISP $\mathcal P_p(m)$ in \eqref{eq:oisp_dist} is \textbf{unique}. Then $\mathcal U_{\mathrm{uniq}}(\Lambda^*)$ is a dense $G_\delta$ subset of $L^p(0,1)$ and, in particular, is residual. Equivalently, the OISP solution is globally unique for a generic reference potential $V_0\in L^p(0,1)$.
	\end{theorem}
	
	\begin{proof}
		1. For $1<p<\infty$, the Banach space $X=L^p(0,1)$ is reflexive, uniformly convex, and uniformly smooth. It has the Radon--Riesz property: if $x_n\rightharpoonup x$ weakly in $X$ and $\|x_n\|_{L^p}\to\|x\|_{L^p}$, then $\|x_n-x\|_{L^p}\to0$.
		
		2. By Lemma~\ref{lem:weak_cont_lp}, the isospectral set $M:=\ISM_m^p(\Lambda^*)\subset L^p(0,1)$ is weakly sequentially closed and hence norm closed. By Theorem~\ref{thm:main_lp} and Lemma~\ref{lem:compact_sol_set}, for every $V_0\in X$ the metric projection
		\begin{equation}\label{eq:metric_proj_def}
			\begin{aligned}
 \Pi_M(V_0)&:=\mathcal S_{\min}(V_0;\Lambda^*)
 =\{\hat V\in M:\ \|\hat V-V_0\|_p=d_M(V_0)\},\\
 d_M(V_0)&:=\inf_{V\in M}\|V-V_0\|_p.
 \end{aligned}
		\end{equation}
		is nonempty and compact.
		
		3. For every integer $k\ge1$, define the set of reference potentials whose metric projection has diameter at least $1/k$:
		\begin{equation}\label{eq:fk_set}
			F_k := \left\{ V_0 \in L^p(0,1) : \diam_{L^p}\big(\Pi_M(V_0)\big) \ge \frac{1}{k} \right\}, \quad \text{where } \diam_{L^p}(A) := \sup_{u, v \in A} \|u - v\|_{L^p}.
		\end{equation}
		The complement of the uniqueness set is the countable union
		\begin{equation}\label{eq:comp_uniq}
			L^p(0,1) \setminus \mathcal{U}_{\mathrm{uniq}}(\Lambda^*) = \bigcup_{k=1}^\infty F_k.
		\end{equation}
		
		4. We prove that every $F_k$ is closed in $L^p(0,1)$. Let $\{V_{0,n}\}_{n=1}^\infty\subset F_k$ and $\|V_{0,n}-V_0\|_{L^p}\to0$. Since $\Pi_M(V_{0,n})$ is compact, there exist $\hat V_n^{(1)},\hat V_n^{(2)}\in\Pi_M(V_{0,n})$ such that $\|\hat V_n^{(1)}-\hat V_n^{(2)}\|_{L^p}\ge\frac1k$. Continuity of the distance function gives $|d_M(V_{0,n})-d_M(V_0)|\le\|V_{0,n}-V_0\|_{L^p}\to0$; hence $\{\hat V_n^{(1)}\}$ and $\{\hat V_n^{(2)}\}$ are bounded in $L^p(0,1)$. Passing to subsequences, $\hat V_n^{(r)}\rightharpoonup\hat V_*^{(r)}$ weakly in $L^p(0,1)$ for $r=1,2$. By Lemma~\ref{lem:weak_cont_lp}, $\hat V_*^{(r)}\in M$, and weak lower semicontinuity of the norm gives
		\[
		d_M(V_0) \le \|\hat{V}_*^{(r)} - V_0\|_{L^p} \le \liminf_{n \to \infty} \|\hat{V}_n^{(r)} - V_{0,n}\|_{L^p} = \lim_{n \to \infty} d_M(V_{0,n}) = d_M(V_0).
		\]
		Thus $\|\hat V_*^{(r)}-V_0\|_{L^p}=d_M(V_0)$, so $\hat V_*^{(1)},\hat V_*^{(2)}\in\Pi_M(V_0)$. Moreover,
\[
 \left|\|\hat V_n^{(r)}-V_0\|_p-d_M(V_0)\right|\le2\|V_{0,n}-V_0\|_p\to0.
\]
The Radon--Riesz property therefore applies to the weakly convergent sequence $\hat V_n^{(r)}-V_0$ with fixed center and gives strong convergence $\hat V_n^{(r)}\to\hat V_*^{(r)}$. Passing to the limit in $\|\hat V_n^{(1)}-\hat V_n^{(2)}\|_{L^p}\ge\frac1k$, we obtain $\|\hat V_*^{(1)}-\hat V_*^{(2)}\|_{L^p}\ge\frac1k$. Hence $V_0\in F_k$, and $F_k$ is closed.
		
		5. We prove that every closed set $F_k$ is nowhere dense in $L^p(0,1)$. Let $V_0\in L^p(0,1)$ and $\varepsilon>0$ be arbitrary, and choose any $\hat V\in\Pi_M(V_0)$. If $V_0=\hat V\in M$, then $\Pi_M(V_0)=\{V_0\}$ and $V_0\notin F_k$. Suppose $V_0\notin M$, so $R:=\|V_0-\hat V\|_{L^p}>0$. Consider an interior point of the segment joining $V_0$ and $\hat V$:
		\begin{equation}\label{eq:steckin_ray}
			V_\theta := (1 - \theta) V_0 + \theta \hat{V}, \quad \theta \in (0, 1), \quad \|V_\theta - V_0\|_{L^p} = \theta R < \varepsilon \quad (\text{if } \theta < \varepsilon / R).
		\end{equation}
		We show that the OISP solution at $V_\theta$ is unique: $\Pi_M(V_\theta)=\{\hat V\}$. For every $V\in M$, the triangle inequality gives
		\[
		\|V - V_\theta\|_{L^p} \ge \|V - V_0\|_{L^p} - \|V_\theta - V_0\|_{L^p} \ge R - \theta R = (1 - \theta) R = \|\hat{V} - V_\theta\|_{L^p}.
		\]
		Thus $d_M(V_\theta)=(1-\theta)R$. If some $\widetilde V\in M$ belongs to $\Pi_M(V_\theta)$, then equality holds in
		\[
		\|\tilde{V} - V_0\|_{L^p} \le \|\tilde{V} - V_\theta\|_{L^p} + \|V_\theta - V_0\|_{L^p} = (1 - \theta)R + \theta R = R
		\]
		In the strictly convex space $L^p(0,1)$, $1<p<\infty$, equality in the triangle inequality for nonzero vectors is possible only when they are positively collinear. Hence $\widetilde V-V_\theta=\frac{1-\theta}{\theta}(V_\theta-V_0)=\hat V-V_\theta$, and $\widetilde V=\hat V$.
		Thus $\Pi_M(V_\theta)=\{\hat V\}$, so $V_\theta\in\mathcal U_{\mathrm{uniq}}(\Lambda^*)$ and $V_\theta\notin F_k$. Since $F_k$ is closed, it is nowhere dense. Therefore $\mathcal U_{\mathrm{uniq}}(\Lambda^*)=\bigcap_{k=1}^\infty(L^p(0,1)\setminus F_k)$ is a dense $G_\delta$ set.
	\end{proof}
	
	\section{Mechanisms of nonuniqueness of OISP solutions}

The generic uniqueness of Theorem~\ref{thm:generic_uniq_baire} allows exceptional centers with several global minima. Smallness in the sense of Baire category alone gives no information about the codimension or measure of the exceptional set.

\begin{example}[Reflection and the essential step in proving multiplicity]\label{ex:parity_breaking}
Consider the reflection
\begin{equation}\label{eq:reflection_op}
 (RV)(x)=V(1-x).
\end{equation}
It is an isometric involution in $L^p(0,1)$. Replacing an eigenfunction $u(x)$ by $u(1-x)$ gives
\begin{equation}\label{eq:reflection_spec}
 \lambda_j(RV)=\lambda_j(V),\qquad R\ISM_m^p(\Lambda^*)=\ISM_m^p(\Lambda^*).
\end{equation}
If $RV_0=V_0$, then
\begin{equation}\label{eq:reflection_norm}
 \|RV-V_0\|_p=\|V-V_0\|_p.
\end{equation}
Thus a nonsymmetric global minimum and its reflection are two distinct global minima. To use this observation, one must exclude \emph{all} symmetric global minima. Reflection alone, like equality of the functional values on two stationary branches, does not provide such an exclusion. The next two examples supply the required comparison for a barrier of finite height.
\end{example}

Set
\begin{equation}\label{eq:double_well_v0}
 J_L=(0,\tfrac14),\quad J_R=(\tfrac34,1),\quad
 C=(\tfrac14,\tfrac34),\qquad V_{0,B}=B\mathbf1_C,
 \qquad \mu=16\pi^2.
\end{equation}

\begin{lemma}[High-barrier limit under weak perturbations]\label{lem:barrier_limit}
Let $1<p<\infty$, $B_n\to\infty$, and $h_n\rightharpoonup h$ in $L^p(0,1)$. Then every fixed eigenvalue of the operator with potential $B_n\mathbf1_C+h_n$ converges to the corresponding eigenvalue of the direct sum of the two Dirichlet operators on $J_L$ and $J_R$ with potentials $h|_{J_L}$ and $h|_{J_R}$. The eigenvalues of the direct sum are counted with multiplicities.
\end{lemma}
\begin{proof}
The norms $\|h_n\|_p$ are bounded. For every $\eta>0$, the one-dimensional interpolation estimate gives, with a common constant,
\[
 \left|\int h_nu^2\right|\le\eta\|u'\|_2^2+C_\eta\|u\|_2^2.
\]
Hence the forms are uniformly bounded from below. An upper bound for the $k$th level is obtained by testing on the first $k$ eigenfunctions of the limiting direct sum, extended by zero on $C$. On their finite-dimensional span, the integrals involving $h_n$ converge uniformly and the barrier term vanishes. Therefore $\limsup\lambda_k\le\nu_k$, where $\nu_k$ is the $k$th limiting level.

Choose the first $k$ orthonormal eigenfunctions $u_{j,n}$ of the original operators. The upper bound for the levels and the lower bound for the forms imply boundedness of their $H_0^1$ norms and
$B_n\int_Cu_{j,n}^2\le C_k$. After extracting a common subsequence, the functions converge weakly in $H_0^1$ and uniformly on $[0,1]$ to orthonormal functions $u_j$ vanishing on $C$. Their restrictions belong to $H_0^1(J_L)\oplus H_0^1(J_R)$. For fixed coefficients $a_j$, set $u_n=\sum_{j=1}^ka_ju_{j,n}$ and $u=\sum_{j=1}^ka_ju_j$. Uniform convergence and weak convergence of $h_n$ give $\int h_nu_n^2\to\int hu^2$. Weak lower semicontinuity of the kinetic energy and nonnegativity of the barrier yield
\[
 a_\infty[u]\le\liminf a_n[u_n]
 =\liminf\sum_{j=1}^k\lambda_j(B_n\mathbf1_C+h_n)a_j^2
 \le(\liminf\lambda_k)\sum_{j=1}^ka_j^2,
\]
where the subsequence was initially chosen to realize the lower limit of the $k$th level. The min--max principle on the span of $u_1,\ldots,u_k$ gives $\nu_k\le\liminf\lambda_k$. The two estimates prove the assertion.
\end{proof}

\begin{example}[A double-well center: at least two minima for $m=1$]\label{ex:double_well}
Fix $1<p<\infty$ and set $i_1=1$, $\lambda_1^*=\mu-1$. For every sufficiently large \emph{finite} $B$, the center in \eqref{eq:double_well_v0} has at least two global OISP solutions.

To prove this, define on the single interval $J_L$
\[
 d=\min\{\|w\|_{L^p(J_L)}:\ \lambda_1^{J_L}(w)=\mu-1\}.
\]
The minimum exists by Theorem~\ref{thm:main_lp} after a linear change of variable; for example, $w=-1$ is admissible. Moreover, $d>0$, since otherwise continuity of the first level at zero would imply $\mu=\mu-1$. Let $w_*$ attain the minimum and let $\widetilde w_*$ be its extension by zero outside $J_L$. By Lemma~\ref{lem:barrier_limit},
$\lambda_1(V_{0,B}+\widetilde w_*)\to\mu-1$, because the first level of the right well equals $\mu$. Set
\[
 c_B=\mu-1-\lambda_1(V_{0,B}+\widetilde w_*),\qquad
 W_B=V_{0,B}+\widetilde w_*+c_B.
\]
A constant shift moves every level by the same amount. Hence $W_B$ is exactly admissible, $c_B\to0$, and $\|W_B-V_{0,B}\|_p^p\to d^p$.

If symmetric global minima $V_n=V_{0,B_n}+h_n$ existed for an unbounded sequence $B_n$, then $\limsup\|h_n\|_p^p\le d^p$. Extract a weak limit $h$, which is symmetric. By Lemma~\ref{lem:barrier_limit}, the first levels of the two identical limiting wells both equal $\mu-1$. By the definition of $d$ and weak lower semicontinuity,
\[
 \liminf\|h_n\|_p^p\ge\|h\|_p^p
 \ge\|h\|_{L^p(J_L)}^p+\|h\|_{L^p(J_R)}^p\ge2d^p,
\]
a contradiction. Thus, for some $B_0(p)$, every global minimum is nonsymmetric whenever $B\ge B_0(p)$. Existence and reflection give at least two distinct solutions. This proof determines neither the exact number of solutions nor the numerical value of $B_0(p)$.
\end{example}

\begin{example}[The same center: at least two minima for $m=2$]\label{ex:double_well_m2}
For fixed $1<p<\infty$, set $i_1=1$, $i_2=2$, and
\[
 \lambda_1^*=\mu-1,\qquad \lambda_2^*=\mu+1.
\]
For every sufficiently large finite $B$, the center in \eqref{eq:double_well_v0} again has at least two global solutions.

Set $s=-\mathbf1_{J_L}+\mathbf1_{J_R}$ and consider $U_{B,a}=V_{0,B}+as$, $a\in[1/2,3/2]$. The first two limiting levels are $\mu-a$ and $\mu+a$; the second levels of the individual wells lie above them. Convergence is uniform in $a$ on the indicated interval. Indeed, by the min--max principle the levels are Lipschitz functions of $a$ with Lipschitz constant one, so pointwise convergence from Lemma~\ref{lem:barrier_limit}, together with a finite grid in $a$, implies uniform convergence. By continuity of the gap, there is $a_B\in(1/2,3/2)$ such that
$\lambda_2(U_{B,a_B})-\lambda_1(U_{B,a_B})=2$; uniform convergence gives $a_B\to1$. Set
\[
 c_B=\mu-1-\lambda_1(U_{B,a_B}),\qquad W_B=U_{B,a_B}+c_B.
\]
Then both spectral equalities hold exactly, $c_B\to0$, and
$\|W_B-V_{0,B}\|_p^p\to\|s\|_p^p=1/2$.

Consequently, the perturbations of global minima remain bounded in $L^p$ as $B\to\infty$. If symmetric global minima existed along an unbounded sequence, their weak limit in Lemma~\ref{lem:barrier_limit} would define two identical wells. The first two levels of the limiting direct sum would then coincide, contrary to the fixed gap $2$. Hence, for every sufficiently large $B$, the global minima are nonsymmetric and occur in reflected pairs.
\end{example}

\begin{remark}[Medial set and stationary branches]\label{ex:caustics}
The set of centers having several nearest points is called the medial set, or the Maxwell set in the corresponding variational interpretation. It is determined by equality of \emph{global} minimal distances. Critical values of the normal map are associated with degeneracy of stationary points, which is a different condition. An intersection of stationary branches, or equality of their values without comparison with all admissible potentials, does not prove global multiplicity.
\end{remark}

\begin{example}[Nonuniqueness of Lagrange multipliers for periodic boundary conditions and graphs with cycles]\label{ex:lame_graph}
		For the Sturm--Liouville operator on an interval with separated boundary conditions, and on every cycle-free metric tree graph \cite{Valeev2026Graph}, the squared eigenfunctions $\varphi_{i_1}^2(\hat V),\dots,\varphi_{i_m}^2(\hat V)$ are always linearly independent by Lemma~\ref{lem:frechet_indep_lp}. Thus each solution $\hat V$ has a unique set of Lagrange multipliers $(\alpha_1,\dots,\alpha_m)\in\mathbb R^m$. Under periodic boundary conditions on a circle, or on a compact graph containing a cycle, their squares may be linearly dependent even when the first $m$ eigenvalues are simple \cite{Valeev2026Graph}. As shown in Section~5.1 of \cite{Valeev2026Graph}, on a circle of length $4K(\kappa)$, $0<\kappa<1$, the Lam\'e potential
		\begin{equation}\label{eq:lame_pot}
			V(x) = 2\kappa^2 \sn^2(x, \kappa)
		\end{equation}
		has the first three simple eigenvalues $\lambda_1=\kappa^2<\lambda_2=1<\lambda_3=1+\kappa^2$. The corresponding eigenfunctions are expressed in terms of Jacobi elliptic functions: $\psi_1(x)=\dn(x,\kappa)$, $\psi_2(x)=\cn(x,\kappa)$, and $\psi_3(x)=\sn(x,\kappa)$. The identities $\cn^2x+\sn^2x=1$ and $\dn^2x+\kappa^2\sn^2x=1$ yield the linear relation
		\begin{equation}\label{eq:lame_dep}
			\dn^2(x, \kappa) - \cn^2(x, \kappa) - (1 - \kappa^2)\sn^2(x, \kappa) \equiv 0.
		\end{equation}
		In full agreement with Theorem~\ref{thm:analytic_pot_dep}, for $c_1=1$, $c_2=-1$, and $c_3=-(1-\kappa^2)$ the first weighted sum in \eqref{eq:Sr_Tr_def} is the nonzero constant $S_1(x)=\kappa^2\dn^2x-\cn^2x-(1-\kappa^4)\sn^2x\equiv-(1-\kappa^2)\ne0$. The Lam\'e potential \eqref{eq:lame_pot} is expressed through squared eigenfunctions by \eqref{eq:V_via_Sr_boxed} and is a real-analytic periodic function. At this point, the Gram matrix of the first three squared eigenfunctions has rank two rather than three, and the Lagrange multipliers are nonunique up to the kernel of the Gram matrix or require a limiting regularized representation; see Theorem~6.1 in \cite{Valeev2026Graph}. This example concerns nonuniqueness of the multipliers for an existing representation and does not by itself prove nonuniqueness of the optimal potential.
	\end{example}
	
	\section{Numerical algorithms for the OISP}

We describe numerical schemes and the conditions for their local applicability. Spectral feasibility and a small stationarity residual do not by themselves certify global minimality; an additional comparison or an analytic certificate is required.

\begin{algorithm}[Frozen-eigenfunction iteration]\label{alg:proj_newton}
Representation~\eqref{eq:lp_rep} is implicit because the eigenfunctions depend on the unknown $V$. The following scheme therefore uses an approximate Jacobian rather than the exact differential of the full system.
\begin{enumerate}
\item On the grid $x_j=jh$, $j=0,\ldots,N+1$, $h=1/(N+1)$, represent the potential by its cell averages
$\overline V_j=h^{-1}\int_{x_j-h/2}^{x_j+h/2}V(x)\,dx$, $j=1,\ldots,N$. This is well defined for every equivalence class $V\in L^p$, independently of the choice of representative. Use the matrix
\begin{equation}\label{eq:jacobi_mat}
 A_h(V)=h^{-2}\tridiag(-1,2,-1)+\diag(\overline V_1,\ldots,\overline V_N).
\end{equation}
An alternative is the finite element method with integration of the potential matrix. On a finite grid, the restriction $i_m\le N$ and preservation of the rank of the discrete constraints must be checked separately.
\item Set $V^{(0)}=V_0$ and $\alpha^{(0)}=0$. At step $s$, compute the selected eigenpairs and normalize them by $h\sum_j\varphi_{i_k,j}^2=1$. Freeze the eigenfunctions at this step.
\item For $p=2$, introduce
\begin{equation}\label{eq:gram_p2}
 G^{(s)}_{kr}=h\sum_j(\varphi_{i_k,j}^{(s)})^2(\varphi_{i_r,j}^{(s)})^2.
\end{equation}
For $p\ne2$, set $r=1/(p-1)$. To define a smoothed scheme explicitly, use
\begin{equation}\label{eq:gram_lp}
 w_\varepsilon(t)=r(t^2+\varepsilon^2)^{(r-1)/2},\qquad
 \psi_\varepsilon(t)=\int_0^t w_\varepsilon(s)\,ds.
\end{equation}
As $\varepsilon\downarrow0$, $\psi_\varepsilon(t)\to t^{\langle r\rangle}$. With the eigenfunctions fixed, the approximate differential of the spectral residual with respect to $\alpha$ is $-G_{p,\varepsilon}^{(s)}$, where
\[
 (G_{p,\varepsilon}^{(s)})_{kr}
 =h\sum_jw_\varepsilon(S_j^{(s)})(\varphi_{i_k,j}^{(s)})^2
 (\varphi_{i_r,j}^{(s)})^2,\qquad
 S_j^{(s)}=\sum_l\alpha_l^{(s)}(\varphi_{i_l,j}^{(s)})^2.
\]
This approximation simultaneously freezes the dependence of the eigenfunctions on $V$. For finite $\varepsilon$, one solves a smoothed system; passage to $\varepsilon\to0$ requires a separate justification.
\item Solve $G_{p,\varepsilon}^{(s)}\Delta\alpha=F_m(V^{(s)})-\Lambda^*$ and choose a damping parameter $\tau_s\in(0,1]$. Then set
\begin{equation}\label{eq:pot_update_lp}
 \begin{aligned}
 \alpha^{(s+1)}&=\alpha^{(s)}+\tau_s\Delta\alpha,\\
 V_j^{(s+1)}&=(1-\tau_s)V_j^{(s)}+
 \tau_s\left[\overline V_{0,j}-\psi_\varepsilon\left(
 \sum_k\alpha_k^{(s+1)}(\varphi_{i_k,j}^{(s)})^2\right)\right].
 \end{aligned}
\end{equation}
For $p=2$, take $\psi_\varepsilon(t)=t$. Monitor the spectral and stationarity residuals separately. Convergence of this iteration is not asserted without additional assumptions.
\end{enumerate}
\end{algorithm}

For comparison, we write the exact implicit differential for $p=2$. Let
\begin{equation}\label{eq:alg_pot_struct}
 R(V,\alpha)=V-V_0+B_V\alpha=0,\qquad B_V\alpha=\sum_k\alpha_k\varphi_{i_k}(V)^2.
\end{equation}
Wherever $H_{V,\alpha}=D_VR$ is invertible, this equation locally determines $V(\alpha)$, and
\begin{equation}\label{eq:exact_schur}
 D_\alpha V=-H_{V,\alpha}^{-1}B_V,\qquad
 D_\alpha F_m(V(\alpha))=-B_V^*H_{V,\alpha}^{-1}B_V.
\end{equation}
The matrix $-G=-B_V^*B_V$ is obtained on the zero section $\alpha=0$ or under the freezing approximation above. Invertibility of the full KKT operator does not require invertibility of $H_{V,\alpha}$ on the whole space; hence the full system is preferable when $V$ cannot be eliminated.

\begin{algorithm}[Homotopy Newton--KKT method for $p=2$]\label{alg:homotopy_kkt}
\begin{enumerate}
\item Introduce the path
\begin{equation}\label{eq:homotopy_path}
 \Lambda^*(t)=(1-t)F_m(V_0)+t\Lambda^*,\qquad 0\le t\le1.
\end{equation}
At $t=0$, the exact pair $(V_0,0)$ is known. Theorem~\ref{thm:tubular_uniq} and invertibility of the KKT operator provide local continuation near the starting point. The branch can then be continued as long as the corresponding KKT operator remains invertible and the solution stays within the range of the local theory. Step-size reduction alone does not preclude degeneracy or a switch of the globally minimizing branch.
\item At each step, solve the residual system
\begin{equation}\label{eq:kkt_residuals}
 r_V=V-V_0+\sum_k\alpha_k\varphi_{i_k}(V)^2,
 \qquad r_\Lambda=F_m(V)-\Lambda^*(t).
\end{equation}
In the discrete inner product $\langle u,v\rangle_h=h\sum_ju_jv_j$, the Newton step is
\begin{equation}\label{eq:kkt_newton_sys}
 \begin{pmatrix}H(V,\alpha)&\Phi_2(V)\\h\Phi_2(V)^T&0\end{pmatrix}
 \begin{pmatrix}\Delta V\\\Delta\alpha\end{pmatrix}
 =-\begin{pmatrix}r_V\\r_\Lambda\end{pmatrix},
\end{equation}
where the columns of $\Phi_2$ are the squares of the normalized eigenvectors. Setting
$P_k=h\varphi_{i_k}\varphi_{i_k}^T$ and
$M_k=A_h(V)-\lambda_{i_k}(V)I+P_k$, we obtain
\begin{equation}\label{eq:kkt_hessian_discrete}
 H(V,\alpha)=I-2\sum_k\alpha_k\diag(\varphi_{i_k})
 (M_k^{-1}-P_k)\diag(\varphi_{i_k}).
\end{equation}
Here $M_k^{-1}-P_k$ is the reduced discrete resolvent. If the selected levels are simple, the constraints have full rank, and the stationary pair is nondegenerate, then the full Newton method converges locally and quadratically from a sufficiently close initial approximation. This local guarantee concerns the stationary pair; minimality must be checked separately.
\end{enumerate}
\end{algorithm}

	\section{Possible applications of OISP theory}
	
	Optimization inverse spectral problems have a broad range of applications both within fundamental mathematics, especially the theory of nonlinear differential equations, and in applied problems of mathematical modeling, acoustics, radiophysics, and vibration mechanics. 
	
\subsection{Nonlinear operators of mathematical physics} The OISP variational method establishes a duality between linear spectral theory and nonlinear differential equations \cite{ValeevIlyasov2018MZ,IlyasovValeev2021}:
	\begin{enumerate}
		\item \textbf{Existence of multisoliton and nodal states for nonvariational NLS systems.} In nonlinear optics and Bose--Einstein condensate physics, propagation of multimode beams in photorefractive crystals and waveguides is described by the system of $m$ incoherently coupled nonlinear Schr\"odinger equations \eqref{eq:p2_nls} \cite{IlyasovValeev2021}. With self-interaction coefficients of different signs, $\sigma_j=\pm1$, system \eqref{eq:p2_nls} does not possess the standard variational structure of an action functional. Reduction to the OISP yields a family of scaled components having $i_k-1$ nodes whenever $\alpha_k\ne0$. Nontriviality of the entire vector is guaranteed by $F_m(V_0)\ne\Lambda^*$; this method does not guarantee that every component is nonzero or that the signs $\sigma_k$ can be prescribed in advance.
		\item \textbf{Generation of solutions with power nonlinearities through an OISP in $L^p(0,1)$.} By Theorem~\ref{thm:main_lp}, a solution of $\mathcal P_p(m)$ in $L^p(0,1)$ yields solutions of the nonlinear system \eqref{eq:lp_nls_sys} with an arbitrary nonlinear exponent $r=\frac1{p-1}\in(0,\infty)$ applied to the quadratic form $\sum_{j=1}^m\sigma_j\hat u_j^2(x)$.
		\item \textbf{Uniqueness of solutions to nonlinear equations.} Convexity of the superlevel set $\mathcal K_{\lambda_1^*}^p$ in Theorem~\ref{thm:uniq_lambda1} provides a geometric proof of uniqueness of the positive solution to the defocusing equation \eqref{eq:oisp_m1_nls}; cf.~\cite{Sadovnichii2026R3,IlyasovSunValeevYao2026}.
	\end{enumerate}
	
	\subsection{Modeling, diagnostics, and spectral design of acoustic and electromagnetic resonators}
	
	In applied physics and engineering, the OISP admits two complementary interpretations \cite{Sadovnichii2026R3,ChuGolub2005,Gladwell2004}:
	\begin{itemize}
		\item \textbf{Identification and nondestructive diagnostics.} Given measured shifts of the first $m$ natural frequencies of a real resonator relative to a reference design $V_0(x)$, one seeks a defect or medium-inhomogeneity profile having the smallest norm of deviation from the reference. A probabilistic interpretation would require a separate statistical model.
		\item \textbf{Spectral design (optimal tuning).} One seeks a minimal modification of the baseline resonator profile $V_0(x)$ that tunes the first $m$ resonance frequencies to the prescribed values $\lambda_1^*,\dots,\lambda_m^*$.
	\end{itemize}
	Since $\varphi_{i_k}^2(x)$ is proportional to the local potential-energy density of the $i_k$th standing wave, formula \eqref{eq:lp_rep} in Theorem~\ref{thm:main_lp} shows that the optimal correction is determined by the joint spatial profile of the squared controlled modes and vanishes at their common nodes. The formula does not imply localization exclusively at individual antinodes.
	
	\paragraph{Acoustic waveguides of variable cross section (Webster's horn equation).}
	Longitudinal acoustic oscillations of a gas in an axisymmetric waveguide or horn with cross-sectional area $S(x)>0$, $x\in[0,1]$, are described by Webster's equation \cite{Eisner1967}:
	\begin{equation}\label{eq:webster_horn}
		-\frac{1}{S(x)} \frac{d}{dx} \left( S(x) \frac{dp(x)}{dx} \right) = \lambda p(x), \quad \lambda = \left(\frac{\omega}{c_0}\right)^2,
	\end{equation}
	where $p(x)$ is the complex amplitude of acoustic pressure, $\omega$ is the angular frequency, and $c_0$ is the speed of sound. The Liouville substitution $u(x)=\sqrt{S(x)}\,p(x)$ transforms Webster's equation \eqref{eq:webster_horn} into the canonical Sturm--Liouville equation \eqref{eq:sl_op} with potential
	\begin{equation}\label{eq:webster_pot}
		-u''(x) + V(x)u(x) = \lambda u(x), \quad V(x) := \frac{\big(\sqrt{S(x)}\big)''}{\sqrt{S(x)}}.
	\end{equation}
	Given the measured first $m$ shifted eigenfrequencies $\lambda_1^*,\dots,\lambda_m^*$, an OISP solution $\hat V(x)$ obtained from \eqref{eq:lp_rep} reconstructs the perturbed potential. A candidate cross section $\hat S=r^2$ can then be found from $r''-\hat Vr=0$ only after two additional data have been prescribed, for example $r(0)=r_0>0$ and $r'(0)=r_1$. Realizability requires $r(x)>0$ on the entire interval and compatibility of the physical boundary conditions with the spectral conditions used. These properties do not follow automatically from the OISP. For instance, when $\hat V=0$, different positive linear functions give different cross sections, whereas for $\hat V=-4\pi^2$ there is no strictly positive solution on all of $[0,1]$. Thus, in a design problem, geometric and positivity constraints must either be included in the formulation or checked separately.
	
	\paragraph{Engineering choice of the norm exponent $p\in(1,\infty)$ in design.}
	Extending OISP theory to $L^p(0,1)$, as in Theorem~\ref{thm:main_lp}, makes it possible to reflect technological considerations through the choice of $p$:
	\begin{enumerate}
		\item \textit{Small exponents $1<p<2$.} Smaller exponents penalize large correction values less strongly than larger exponents. This may motivate their use in modeling localized changes, but does not prove sparsity or concentration of the optimal profile. Concentrated elements require a separate model allowing measures.
		\item \textit{The Hilbert case $p=2$.} This is natural for smooth distributed corrections of elastic or acoustic medium parameters with minimal perturbation energy.
		\item \textit{Large exponents $p\gg2$.} In manufacturing optical waveguides and microwave resonators, the maximum change in refractive index or etching depth is limited by a technological tolerance $|\hat V(x)-V_0(x)|\le\Delta_{\max}$. Large finite exponents penalize large amplitudes more strongly, but guarantee neither exact minimization of the $L^\infty$ norm nor satisfaction of the prescribed tolerance. The hard constraint $|V-V_0|\le\Delta_{\max}$ should be imposed explicitly; existence of an admissible potential and optimality conditions for that problem require a separate analysis.
	\end{enumerate}
	
	\section{Conclusion}
	
	We have studied optimization inverse spectral problems for Sturm--Liouville operators throughout the Banach scale $L^p(0,1)$, $1<p<\infty$. The main results are as follows:
	\begin{enumerate}
		\item For every $1<p<\infty$, every reference potential $V_0\in L^p(0,1)$, and every strictly increasing finite set of eigenvalues $\Lambda^*=(\lambda_1^*,\dots,\lambda_m^*)^T$, we proved existence of a global OISP solution $\hat V\in\ISM_m^p(\Lambda^*)$; see Theorem~\ref{thm:main_lp}.
		\item We obtained an explicit analytic representation of the optimal potential in $L^p(0,1)$ through the dual power map applied to a linear combination of squared eigenfunctions,
		$\hat V(x)=V_0(x)-\left(\sum_{k=1}^m\alpha_k\varphi_{i_k}^2(\hat V)(x)\right)^{\langle\frac1{p-1}\rangle}$.
		For $p=2$, this reduces to the classical Ilyasov--Valeev representation \cite{IlyasovValeev2021}. We also derived the corresponding system of nonlinear differential equations with nonlinear exponent $\frac1{p-1}$.
		\item Using a unified method of differential identities for weighted sums of squared eigenfunctions and their derivatives, we proved that linear dependence of the squared eigenfunctions forces the potential to be expressible through those squares and to admit a real-analytic representative; see Lemma~\ref{lem:Sr_Tr_identities} and Theorem~\ref{thm:analytic_pot_dep}.
		\item For every $1<p<\infty$, we proved compactness of the set of global minimizers. For $p=2$, we established isolation of Morse nondegenerate stationary points and finiteness of the number of global minima when each of them is nondegenerate; see Lemma~\ref{lem:compact_sol_set} and Theorem~\ref{thm:isolated_finite}.
		\item We established unconditional global uniqueness of the OISP solution in $L^p(0,1)$ for the first eigenvalue when $\lambda_1^*\ge\lambda_1(V_0)$, uniqueness in an open neighborhood of the isospectral manifold for $p=2$, and generic global uniqueness in the Baire sense throughout $L^p(0,1)$ for arbitrary $m\ge1$; see Theorems~\ref{thm:uniq_lambda1}, \ref{thm:tubular_uniq}, and \ref{thm:generic_uniq_baire}.
		\item For $m=1,2$, we proved existence of at least two global minima in the presence of a sufficiently high finite symmetric barrier. We described numerical schemes with local applicability conditions and possible applications; no global convergence of the algorithms or physical realizability of an arbitrary optimal potential is asserted.
	\end{enumerate}
	

\end{document}